\documentclass{article}
\usepackage{mathtools}
\usepackage{amsmath,amsthm,amssymb}
\usepackage[dvipsnames,table]{xcolor}
\usepackage[noadjust]{cite}
\usepackage{enumitem}
\usepackage{xurl}
\usepackage[margin=2.5cm]{geometry}
\usepackage{graphicx}
\usepackage[T1]{fontenc}
\usepackage{authblk}
\usepackage[english]{babel}

\usepackage{hyperref}
\hypersetup{
  colorlinks   = true,
  urlcolor     = blue,
  linkcolor    = Purple,
  citecolor    = red,
  breaklinks   = true
}
\usepackage[capitalize,noabbrev]{cleveref}

\theoremstyle{definition}
\newtheorem{theorem}{Theorem}[section]
\newtheorem{proposition}[theorem]{Proposition}
\newtheorem{corollary}[theorem]{Corollary}
\newtheorem{lemma}[theorem]{Lemma}
\newtheorem{definition}[theorem]{Definition}
\newtheorem{notation}[theorem]{Notation}
\newtheorem{remark}[theorem]{Remark}
\newtheorem{conjecture}[theorem]{Conjecture}
\newtheorem*{theorem*}{Theorem}
\usepackage{thmtools}
\usepackage{thm-restate}

\newcommand{\F}{\mathbb F}

\newcommand{\Fq}{\F_q}

\newcommand{\D}{\mathcal D}

\newcommand{\C}{\mathcal C}

\DeclareMathOperator{\rk}{rk}

\DeclareMathOperator{\supp}{supp}
\DeclareMathOperator{\wt}{wt}

\DeclareMathOperator{\rs}{rowsp}
\DeclareMathOperator{\cs}{colsp}
\DeclareMathOperator{\Gr}{Gr}
\DeclareMathOperator{\vol}{vol}

\newcommand{\ceil}[1]{\left\lceil#1\right\rceil}

\title{A proof of the generalized packing--covering conjecture}

\author[1]{Gianira N. Alfarano}
\author[2]{Giuseppe Marino}
\author[2]{Alessandro Neri}
\author[2]{Rocco Trombetti}

\affil[1]{Universit\'e de Rennes, IRMAR}
\affil[2]{Department of Mathematics and Applications ``R. Caccioppoli'', University of Naples Federico II}

\affil[ ]{{\small \texttt{gianira-nicoletta.alfarano@univ-rennes.fr},
\texttt{giuseppe.marino@unina.it},}}\affil[ ]{{\small
 \texttt{alessandro.neri@unina.it},
\texttt{rocco.trombetti@unina.it}}}
\date{}

\begin{document}

\maketitle
\begin{abstract}
The generalized packing--covering conjecture of Elimelech, Firer and
Schwartz asserts that, for every linear code $\C$ and every admissible
order $t$, the $t$-th generalized Hamming weight $d_t(\C)$ and the $t$-th generalized covering radius $R_t(\C)$ satisfy $d_t(\C)\le 2R_t(\C)+2$. We give a computer-assisted proof of the conjecture for every linear code over every finite field and every admissible order. Combining a parity-check reformulation of the conjecture, bounds on the length of putative counterexamples, and successive puncturing arguments, we settle all orders $t\ge 32$ and reduce the remaining orders to finitely many parameter tuples, which we exclude by an exact computer verification.
\end{abstract}



\section{Introduction}
 
The minimum distance and the covering radius are two of the most basic parameters of an error-correcting code. For a linear code
$\C\subseteq\Fq^n$ with minimum distance $d$, the Hamming balls of radius $\lfloor (d-1)/2\rfloor$ centred at the codewords are pairwise disjoint, while the covering radius $R(\C)$ is the smallest radius for which the balls centred at the codewords cover the whole space. A classical elementary argument shows that 
$$\left\lfloor\frac{d-1}{2}\right\rfloor\le R(\C),$$
that is, the packing radius of a code never exceeds its covering radius;
see, e.g.,~\cite{HuffmanPless,CHLL}. We refer to~\cite{CKMS,CHLL} for
surveys on the covering radius. Both parameters have natural higher-order
generalizations.
 
Generalized Hamming weights were introduced by Wei~\cite{Wei}, motivated by the wire-tap channel of type~II \cite{OW}, although related notions had already appeared in~\cite{HKM,Klove}. The \emph{$t$-th generalized Hamming weight} $d_t(\C)$ is the smallest size of the support of a $t$-dimensional subcode of $\C$. Wei proved that these weights are strictly increasing in $t$ and
that the weight hierarchies of a code and of its dual determine each other.
Independently, Forney showed that they can be used to describe the trellis complexity of linear codes~\cite{Forney}. Since then, generalized Hamming weights
have been studied from several points of view: geometrically, in terms of the projective systems associated with a code~\cite{TV}; combinatorially, through matroids and their Betti numbers~\cite{Britz,JV}; and for many
explicit families, including cyclic codes~\cite{FTW} and Reed--Muller
codes~\cite{HP}, as well as, more recently, affine Cartesian
codes~\cite{BD}, weighted projective Reed--Muller codes~\cite{nardi2026structure} and codes arising from complete
intersections~\cite{moreno2026generalized}. Computing generalized Hamming weights is hard in general, since already the computation of the minimum distance is NP-hard~\cite{Vardy}; algorithms and an implementation in SageMath for computing the whole weight hierarchy of a code are given in~\cite{GHWs}. Classical bounds such as the Griesmer bound~\cite{Griesmer} have also been extended to generalized weights~\cite{Deng}.

Generalized covering radii were introduced much more recently by Elimelech, Firer and Schwartz~\cite{EFS}. The quantity $R_t(\C)$ is the smallest integer $r$ such that, for any $t$ vectors of $\Fq^n$, there is a set of at most $r$ coordinates outside which each vector agrees with some codeword of $\C$.
Equivalently, $R_t(\C)$ is the ordinary covering radius of
the extended code $\C\otimes\F_{q^t}$ over $\F_{q^t}$~\cite{EFS}. These invariants have been studied for Reed--Muller codes~\cite{EWS} and BCH codes~\cite{ozbudak2026third, xiong2026asymptotic}. Moreover, a geometric description in terms of saturating systems of subspaces was given in~\cite{AMNT}.
 
In analogy with the classical case, Elimelech, Firer and Schwartz
defined the \emph{$t$-th generalized packing radius}
$\delta_t(\C)=\lfloor(d_t(\C)-1)/2\rfloor$ and conjectured that
$\delta_t(\C)\le R_t(\C)$ for every admissible $t$~\cite{EFS}; see
Conjecture~\ref{conj:packing-covering}. It is worth stressing that the
conjecture does not follow from the classical inequality applied to
$\C\otimes\F_{q^t}$: this code has covering radius $R_t(\C)$, but its
minimum distance is still $d_1(\C)$, not $d_t(\C)$. In terms of a
parity-check matrix, $R_t(\C)$ describes how many columns are needed to span every $t$-dimensional subspace of the syndrome space, while $d_t(\C)$ describes how few columns are needed to have nullity $t$.
 
Several cases of the conjecture have been established. The case $t=1$ recovers
the classical inequality for packing and covering radii. The case $t=2$, as
well as a number of parameter ranges for arbitrary $t$ (including $R_t(\C)=t$, codes of small redundancy relative to $t$, and codes of rate at most $3/5$), were settled by Yu and Schwartz~\cite{YS}. More recently, Essayag and Zabokritskiy~\cite{EZ} proved the conjecture for all codes of
redundancy at most 14, and established the stronger bound $d_t(\C)\le 2R_t(\C)+1$ whenever $q\ge R_t(\C)$, together with a few small cases of $(t,R_t(\C))$.

In this paper we give a computer-assisted proof of the generalized
packing--covering conjecture in full.
\begin{restatable}{thm}{MioTeorema}
\label{thm:main}
Let $\C$ be an $[n,k]_q$ linear code and let
$1\le t\le\min\{k,n-k\}$. Then
$$
d_t(\C)\le 2R_t(\C)+2.
$$
\end{restatable}
 
\subsection*{Outline of the proof}
The proof is by contradiction. Suppose that $\C$ is an $[n,n-\rho]_q$ code such that $d_t(\C)\ge 2r+3$, where $r=R_t(\C)$. For
such a code we prove two bounds on the length~$n$.

\begin{itemize}
\item \emph{A lower bound.} By the ball-covering bound \eqref{eq:cover},
  the Hamming ball of radius~$r$ in $\F_{q^t}^n$ has at least
  $q^{t\rho}$ elements. This forces $n$ to be large: roughly,
  $n$ grows like $q^{t(\rho-r)/r}$, by \eqref{eq:coverrough}.
\item \emph{An upper bound.} Since $d_t(\C)$ is large, every subspace of
  $\F_q^\rho$ of small dimension contains only few columns of a
  parity-check matrix of $\C$; see \eqref{eq:sparsity}. Counting columns
  subspace by subspace gives $n\le M(q,\rho,t,r)$, where $M$ is an
  explicit function; see Section~\ref{sec:putative_code}.
\end{itemize}

If the lower bound exceeds $M(q,\rho,t,r)$, no such code exists. This
already excludes many parameters, but not all of them. In the remaining
cases we use the lower bound in a different way. If $n$ is large, the
Hamming ball of a suitable radius $\ell$ in $\F_q^n$ has more than
$q^{\rho}$ elements, that is, more elements than there are syndromes.
Hence, many vectors of the ball have the same syndrome, and their
differences are codewords of weight at most $2\ell$. More precisely,
Lemma~\ref{lem:hamming-ball} gives an $s$-dimensional subcode of $\C$
supported on at most $(s+1)\ell$ coordinates. One such step rarely
gives a $t$-dimensional subcode on few enough coordinates. We therefore
puncture the support found, which lowers the redundancy, and repeat the
argument on the shorter code; Theorem~\ref{thm:successive}. If the
steps together produce a $t$-dimensional subcode supported on at most
$2r+2$ coordinates, then $d_t(\C)\le 2r+2$, a contradiction.
Section~\ref{sec:reduction} uses these two tools to exclude all
parameters $(q,\rho,t,r)$, one region at a time.
\begin{enumerate}[label=\textbf{Step \arabic*.},leftmargin=*,align=left]
\item \emph{Small redundancy and small gaps.} We first treat all codes
  with $\rho\le50$ by an exact computation; see Theorem~\ref{thm:finite}.
  For fixed $t$ and $r$, a single Hamming ball settles all sufficiently
  large redundancies (Theorem~\ref{thm:tail}), and a similar argument
  settles all large orders when the gap $r-t$ is fixed; see Theorem~\ref{thm:diagonal}. Together these results leave finitely
  many cases for each gap, and we check the gaps $r-t\le5$ by computer; see Theorem~\ref{thm:fivegaps}.
\item \emph{Orders $t\ge32$.} For large orders, we show that the
  redundancy of a counterexample exceeds twice its covering radius. The
  code is then long enough that the successive puncturing succeeds with
  an explicit choice of radii, and we give such a choice separately for
  $r\le2t$, $2t\le r\le3t$ and $r\ge3t$. The only computer check needed
  here concerns finitely many pairs $(t,r)$ with $6\le r-t\le9$; see Theorem~\ref{thm:highorder}.
\item \emph{Orders $3\le t\le31$.} For small orders the constants of
  Step~2 are too weak, and we prove sharper estimates for Hamming balls
  after puncturing; see Section~\ref{sec:low}. They give, for each such order, an explicit bound on the covering radius of a counterexample; see Theorems~\ref{thm:tails}, \ref{thm:ternary} and~\ref{thm:binarytail}.
  The order $t=3$ is the tightest case and needs separate arguments.
\item \emph{The remaining finite set.} After Steps~1--3, the parameters
  $q$, $\rho$, $t$ and $r$ of a counterexample lie in an explicit finite
  set. For each tuple in it, we either show that the two bounds on $n$ are
  incompatible, or check the hypotheses of Theorem~\ref{thm:successive}
  at a fixed test length; Lemma~\ref{lem:shorten} then covers every larger length; see Theorems~\ref{thm:extensions}
  and~\ref{thm:completion}. The proof of Theorem~\ref{thm:main}
  collects all the steps.
\end{enumerate}

Table~\ref{tab:roadmap} shows where each region of parameters is
treated. The computations are described in
Appendix~\ref{app:computations}; they only check explicit integer or
rational inequalities and never enumerate codes.

\begin{table}[ht]
\centering\small
\begin{tabular}{ll}
\hline
Parameters & Treated in\\
\hline
$t\le2$, $r=t$, $q\ge r$, and other known cases & Proposition~\ref{prop:inputs}\\
$\rho\le50$ & Theorem~\ref{thm:finite}\\
$\rho\ge D(t,r)$ & Theorem~\ref{thm:tail}\\
$r-t\le5$ & Theorem~\ref{thm:fivegaps}\\
$t\ge32$ & Theorem~\ref{thm:highorder}\\
$3\le t\le31$, $r$ large & Theorems~\ref{thm:tails}, \ref{thm:ternary}, \ref{thm:binarytail}\\
$r\le2t$ or $r-t\le15$ & Theorem~\ref{thm:extensions}\\
all remaining cases & Theorem~\ref{thm:completion}\\
\hline
\end{tabular}
\caption{Where each range of parameters is settled.}
\label{tab:roadmap}
\end{table}

\section{Preliminaries}\label{sec:prelim}

In this section we provide the necessary background material for the rest of the paper. We first fix the notation and recall the notions of \emph{generalized Hamming weights} and \emph{generalized covering radii}, together with the \textbf{generalized packing--covering conjecture}.

\begin{notation}
Let $q$ be a prime power, and let $\Fq$ denote the finite field with $q$ elements. For a positive integer $n$, write $[n]:=\{1,\ldots,n\}$.
For a matrix $A\in\Fq^{h\times n}$ and a subset $I\subseteq[n]$, we denote by $A_I\in\Fq^{h\times |I|}$ the submatrix of $A$ obtained by restricting to the columns indexed by $I$. We write $\rs_{\Fq}(A)$ and $\cs_{\Fq}(A)$ for the $\Fq$-span of the rows and, respectively, of the columns of $A$. In the paper, $\mathrm{e}$ denotes the Euler's constant. 
\end{notation}

\subsection{Linear codes, generalized Hamming weights and covering radii}

In this subsection we recall the basic notions on linear codes, generalized Hamming weights and generalized covering radii that will be used throughout the paper. For further background on linear codes we refer to~\cite{HuffmanPless}, while for generalized Hamming weights and generalized covering radii we refer to \cite{Wei,EFS}; see also \cite{TV,AMNT}.

Let $v\in\Fq^n$. The \textbf{support} of $v$ is
$\supp(v):=\{i\in[n]:v_i\ne0\}$ and its \textbf{Hamming weight} is
$\wt(v):=|\supp(v)|$. The \textbf{Hamming distance} between two vectors
$u,v\in\Fq^n$ is defined as $d_{\rm H}(u,v):=\wt(u-v)$.

An $[n,k]_q$ \textbf{linear code} is a $k$-dimensional $\Fq$-subspace
$\C\subseteq\Fq^n$. Its elements are called \textbf{codewords}, and its
\textbf{redundancy} is $\rho:=n-k$. The support of $\C$ is
$$
\supp(\C):=\bigcup_{c\in\C}\supp(c).
$$
The \textbf{minimum distance} of $\C$ is
$d(\C):=\min\{\wt(c):c\in\C,\ c\ne0\}$.

A matrix $G\in\Fq^{k\times n}$ is called a \textbf{generator matrix} of
$\C$ if $\rs_{\Fq}(G)=\C$. A full-row-rank matrix
$H\in\Fq^{\rho\times n}$ is called a \textbf{parity-check matrix} of $\C$
if $\C=\ker(H)$, that is,
$$
\C=\{c\in\Fq^n:Hc^\top=0\}.
$$
Equivalently, $\rs_{\Fq}(H)=\C^\perp$, where $\C^\perp$ denotes the
\textbf{dual code} of $\C$, with respect to the standard inner product.

For a fixed parity-check matrix $H$, the map $v\mapsto Hv^\top$ from
$\Fq^n$ to $\Fq^\rho$ is called the \textbf{syndrome map}, and the vectors
of $\Fq^\rho$ are called \textbf{syndromes}. Since $H$ has full row-rank,
every vector of $\Fq^\rho$ occurs as a syndrome.

We do not assume that the codes considered in this paper are
nondegenerate: zero coordinates are allowed throughout.

Generalized Hamming weights were introduced by Wei~\cite{Wei}; see also~\cite{TV} for their geometric interpretation.

\begin{definition}\label{def:ghw}
Let $\C$ be an $[n,k]_q$ linear code and let $t\in[k]$. The $t$-th
\textbf{generalized Hamming weight} of $\C$ is defined as
$$
d_t(\C):=\min\{|\supp(\D)|:\D\subseteq\C,\ \dim_{\Fq}(\D)=t\}.
$$
\end{definition}

In particular, $d_1(\C)=d(\C)$. We will repeatedly use the well-known \emph{generalized
Singleton bound} \cite[Theorem~1]{Wei}, which states that for every $t\in [k]$, $$d_t(\C)\le\rho+t.$$

We next recall generalized covering radii, introduced by Elimelech, Firer,
and Schwartz in \cite{EFS}. We use the formulation in terms of a parity-check
matrix, which is the one that will be most convenient throughout the paper. A geometric treatment of generalized covering radii was given in \cite{AMNT}.

\begin{definition}\label{def:gcr}
Let $\C$ be an $[n,k]_q$ linear code with parity-check matrix
$H\in\Fq^{\rho\times n}$, and let $t\in[\rho]$. The $t$-th
\textbf{generalized covering radius} of $\C$ is defined as
$$
R_t(\C):=
\max_{\substack{S\subseteq\Fq^\rho\\ |S|=t}}
\min\{|I|:I\subseteq[n],\ S\subseteq\cs_{\Fq}(H_I)\}.
$$
\end{definition}

The definition does not depend on the choice of the parity-check matrix.
When $t=1$, it coincides with the classical covering radius of $\C$. Moreover, it is easy to see that $$0\le R_1(\C)\leq R_2(\C)\le \cdots \le R_{\rho}(\C)=\rho.$$

Several equivalent definitions of generalized covering radii were established
in \cite{EFS}; see also \cite[Proposition~2.4]{AMNT}. In particular,
$R_t(\C)$ is the smallest integer $r$ such that, for every
$v_1,\ldots,v_t\in\Fq^n$, there exist $c_1,\ldots,c_t\in\C$ and a set
$I\subseteq[n]$ with $|I|\le r$ such that
$\supp(v_j-c_j)\subseteq I$ for every $j\in[t]$.

Equivalently, $R_t(\C)$ is the classical covering radius of the code
$\C\otimes\F_{q^t}$ over $\F_{q^t}$; see \cite[Lemma~7]{EFS}.

The connection with generalized Hamming weights led Elimelech, Firer, and
Schwartz to introduce the $t$-th \textbf{generalized packing radius}
$$\delta_t(\C):=\lfloor(d_t(\C)-1)/2\rfloor$$
and to formulate the following
conjecture; see \cite[Section~VI]{EFS}.

\begin{conjecture}[Generalized packing--covering conjecture]
\label{conj:packing-covering}
Let $\C$ be an $[n,k]_q$ linear code and let
$t\in [\min\{k,n-k\}]$. Then
$$\left\lfloor\frac{d_t(\C)-1}{2}\right\rfloor\le R_t(\C).$$
Equivalently,
$$d_t(\C)\le2R_t(\C)+2.$$
\end{conjecture}

For $t=1$, Conjecture~\ref{conj:packing-covering} is the classical inequality
between the packing and covering radii. The purpose of this paper is to prove
the conjecture for every admissible order.

Throughout the paper we assume $ t\in[\min\{k,\rho\}]$ and write
$r:=R_t(\C)$.
We will repeatedly use that $t\le r\le\rho$. Indeed, $t$ linearly independent syndromes cannot be
contained in the span of fewer than $t$ columns of a parity-check matrix, while the full syndrome space is generated by $\rho$ independent columns.

\subsection{Known cases of the conjecture and the covering bound}

Several cases of Conjecture~\ref{conj:packing-covering} are already known. We collect here the results that will be used later.

\begin{proposition}\label{prop:inputs}
Let $\C$ be an $[n,k]_q$ linear code with redundancy $\rho=n-k$,
let $t\in[\min\{k,\rho\}]$, and set $r:=R_t(\C)$.
Then $d_t(\C)\le2r+2$ in each of the following cases:
\begin{enumerate}[label=(\arabic*)]
\item $t\le2$; $r=t$; $t\ge\rho-4$; $k\le5t-2$; $3r\ge2\rho$; or $k/n\le3/5$;
\item $q\ge r$; or $(t,r)\in\{(3,4),(3,5),(4,5),(4,6)\}$.
\end{enumerate}
If $q\ge r$, the stronger bound $d_t(\C)\le2r+1$ holds.
\end{proposition}

\begin{proof}
The case $t=1$ is classical. The case $t=2$ has been proved in \cite[Theorem~5]{YS}. The cases $r=t$, $t\ge\rho-4$, $k\le5t-2$, $3r\ge2\rho$, and $k/n\le3/5$ are given, respectively, in \cite[Lemmas~6, 7, 9, 11 and Theorem~10]{YS}. The claims in the second item follow from \cite[Theorem~1.2 and Corollary~3.2]{EZ}.
\end{proof}

In particular, if $t\ge3$ and $d_t(\C)\ge2R_t(\C)+3$, then
\begin{equation}\label{eq:rate}
k\ge5t-1\qquad \textnormal{and} \qquad n>\frac{5\rho}{2}.
\end{equation}

We also recall the \emph{ball-covering bound} for generalized covering radii. Let $Q\ge2$ and let $N$ be a positive integer. For $x\in\F_Q^N$ and an integer $\ell\ge0$, the \textbf{Hamming ball} of radius $\ell$ and center $x$ is
$$
B_\ell(x):=\{y\in\F_Q^N:d_{\rm H}(x,y)\le \ell\}.
$$
Its cardinality does not depend on the center $x$ and is given by
$$
\vol_Q(N,\ell):=|B_\ell(x)|
=\sum_{i=0}^{\min\{N,\ell\}}\binom Ni(Q-1)^i.
$$
We refer to $\vol_Q(N,\ell)$ as the \textbf{volume} of a Hamming ball of radius $\ell$ in $\F_Q^N$.

The \emph{generalized ball-covering bound} from \cite{EFS} states that
\begin{equation}\label{eq:cover}
\vol_{q^t}(n,r)\ge q^{t\rho}.
\end{equation}
See also \cite[Equation~(4)]{YS}. Indeed, a $t$-tuple of vectors in $\Fq^n$ can be viewed columnwise as a word of length $n$ over an alphabet of size $q^t$. Every $t$-tuple of syndromes has representatives with common support of size at most $r$; viewed as words in $\F_{q^t}^n$, these representatives lie in the Hamming ball of radius $r$ centered at zero. Distinct $t$-tuples of syndromes have distinct representatives, and there are $q^{t\rho}$ possible $t$-tuples of syndromes.

As shown in \cite[Theorem 4.2]{AMNT}, since $\vol_Q(n,r)\le\binom nr Q^r$, \eqref{eq:cover} implies that
\begin{equation}\label{eq:coverrough}
\binom nr\ge q^{t(\rho-r)}\qquad \textnormal{and}\qquad
n>\frac r{\mathrm e}q^{t(\rho-r)/r}.
\end{equation}
Here $\mathrm e$ is Euler's number. For the second
inequality, we used $r!>(r/\mathrm e)^r$ and
$\binom nr\le n^r/r!$.

The estimates in \eqref{eq:cover} and \eqref{eq:coverrough} will be used to provide lower bounds on the length of a putative code violating the conjecture.

\section{Putative codes violating the conjecture}\label{sec:putative_code}
In this section we study the consequences of assuming that a code violates Conjecture~\ref{conj:packing-covering}. We first translate this assumption into a condition on the columns of a parity-check matrix. We then use this condition to derive upper bounds on the length of a counterexample, and finally develop an argument based on Hamming balls and puncturing, which will be used to exclude the remaining cases.

Throughout this section, let $\C$ be an $[n,k]_q$ linear code with
redundancy $\rho=n-k$, and fix a parity-check matrix
$H=(h_1,\ldots,h_n)\in\Fq^{\rho\times n}$ of $\C$.
We regard the columns of $H$ as a multiset of $n$ vectors. Thus, zero columns are allowed, and repeated or proportional columns are always counted with their multiplicities.

Finally, recall that \textbf{puncturing} $\C$ on $I\subseteq[n]$ gives the code
obtained by deleting the coordinates indexed by $I$ from every
codeword of $\C$. \textbf{Shortening} $\C$ on $I$ gives the code obtained
by first restricting to the codewords $c\in\C$ with $c_i=0$ for
every $i\in I$, and then deleting the coordinates indexed by $I$.

\subsection{Parity-check formulation and nullity}\label{sec:parity}

We now translate generalized Hamming weights and generalized covering radii into properties of the columns of a parity-check matrix. We use the following notation.

\begin{notation}
    For a subspace $W\subseteq\Fq^\rho$, let
$$N(W):=|\{i\in[n]:h_i\in W\}|$$
and $$\mu_s:=\max_{\dim (W)=s}N(W).$$ For $I\subseteq[n]$, write
$r_H(I):=\rk(H_I)$ and $\eta(I):=|I|-r_H(I)$.
\end{notation}

The quantity $\eta(I)$ is the \textbf{nullity} of the set of columns indexed by $I$. In other words, it is the dimension of the subcode of $\C$ consisting of codewords whose support is contained in $I$. The standard parity-check description of generalized Hamming weights therefore gives
\begin{equation}\label{eq:weights}
d_t(\C)=\min\{|I|:I\subseteq[n],\ |I|-\rk (H_I)\ge t\}.
\end{equation}
See, for instance, \cite{Wei,TV}.

For generalized covering radii, Definition~\ref{def:gcr} can be expressed entirely in terms of subspaces of the syndrome space. The following is the subspace-covering criterion of \cite[Lemma~3.4]{AMNT}, combined with the correspondence between generalized covering radii and saturating sets given in \cite[Theorem~3.3]{AMNT}.

\begin{proposition}\label{prop:saturation}
For $t\le r\le\rho$, the following are equivalent:
\begin{enumerate}[label=(\arabic*)]
\item $R_t(\C)\le r$;
\item every $t$-dimensional subspace $T\subseteq\Fq^\rho$ is contained in $\cs_{\Fq}(H_I)$ for some $I\subseteq[n]$ with $|I|\le r$.
\end{enumerate}
\end{proposition}
Equivalently, denoting by $\Gr_q(t,\rho)$ the Grassmannian of $t$-dimensional subspaces of $\Fq^\rho$,
\begin{equation}\label{eq:radii}
R_t(\C)=
\max_{T\in\Gr_q(t,\rho)}
\min\{|I|:I\subseteq[n],\ T\subseteq\cs_{\Fq}(H_I)\}.
\end{equation}

\begin{remark}
Combining \eqref{eq:weights} and \eqref{eq:radii} gives an equivalent form of the conjecture that we will use in the rest of the paper.
If $r=R_t(\C)$, then Conjecture~\ref{conj:packing-covering} asks for a set $I\subseteq[n]$ such that $|I|\le2r+2$ and $\eta(I)\ge t$.
In this formulation, the two parameters in the conjecture are described in terms of the columns of a parity-check matrix. The generalized covering radius describes how many columns are needed to span each $t$-dimensional subspace of the syndrome space, while the generalized Hamming weight measures how few columns are needed to obtain nullity at least $t$.
For an integer $m\ge0$, \eqref{eq:weights} shows that $d_t(\C)\le m+t$ means that some set of at most $m+t$ columns
of $H$ has rank at most its size minus $t$.
\end{remark}

The next lemma shows that this is the same as asking for an $m$-dimensional subspace of $\Fq^\rho$ that contains at least $m+t$ columns of $H$.

\begin{lemma}\label{lem:density}
For $0\le m\le\rho$ and $t\in[k]$, we have
$d_t(\C)\le m+t$ if and only if $\mu_m\ge m+t$.
\end{lemma}

\begin{proof}
If an $m$-dimensional subspace contains at least $m+t$ columns of $H$, counted with multiplicity, choose any $m+t$ of them. Their rank is at most $m$, so their nullity is at least $t$, and \eqref{eq:weights} gives $d_t(\C)\le m+t$.

Conversely, suppose $d_t(\C)\le m+t$. By \eqref{eq:weights}, there is a set $I\subseteq[n]$ with $|I|\le m+t$ and $u:=r_H(I)\le |I|-t\le m$. Since the columns of $H$ span $\Fq^\rho$ and $u\le m\le\rho$, we may append $m-u$ columns, each increasing the rank by one. The span $W$ of all these columns has dimension $m$ and contains at least
$|I|+m-u\ge m+t$
columns of $H$, counted with multiplicity. Hence, $\mu_m\ge N(W)\ge m+t$.
\end{proof}

We now apply Lemma~\ref{lem:density} to the bound in
Conjecture~\ref{conj:packing-covering}. Set
$$
m:=2r-t+2,
$$
so that $m+t=2r+2$. If $m\ge\rho$, then $\rho+t\le2r+2$, and the
generalized Singleton bound \cite[Theorem~1]{Wei} gives
$d_t(\C)\le\rho+t\le2r+2$. We may therefore assume $m<\rho$.

In this range, Lemma~\ref{lem:density} gives
$$
d_t(\C)\le2r+2
\quad\Longleftrightarrow\quad
\mu_{2r-t+2}\ge2r+2.
$$
Thus, Conjecture~\ref{conj:packing-covering} is equivalent to proving that
some $(2r-t+2)$-dimensional subspace of $\Fq^\rho$ contains at least
$2r+2$ columns of the parity-check matrix $H$, counted with multiplicity.

From now on, whenever we assume that $\C$ is a counterexample, we therefore have
\begin{equation}\label{eq:failure}
d_t(\C)\ge2r+3.
\end{equation}
In this case we say that $\C$ is a \textbf{counterexample for order $t$}.
Recall from Section~\ref{sec:prelim} that $t\le r$; moreover, the generalized Singleton bound gives
\begin{equation}\label{eq:failure-range}
t\le r\le
\left\lfloor\frac{\rho+t-3}{2}\right\rfloor
\qquad \textnormal{and} \qquad
\rho\ge2r-t+3.
\end{equation}
We also set
\begin{equation}\label{eq:gamma}
\gamma:=\rho-m=\rho-2r+t-2.
\end{equation}

By Lemma~\ref{lem:density}, \eqref{eq:failure} implies $\mu_m\le m+t-1=2r+1$. More generally, every subspace $W\subseteq\Fq^\rho$ with $\dim (W)\le m$ satisfies
\begin{equation}\label{eq:sparsity}
N(W)\le\dim (W)+t-1.
\end{equation}
Indeed, if $N(W)\ge\dim (W)+t$, we may enlarge $W$ to an $m$-dimensional subspace by successively adding columns of $H$ that increase its dimension. The resulting $m$-dimensional subspace contains at least $m+t$ columns of $H$, counted with multiplicity, contradicting $\mu_m\le m+t-1$.

Condition~\eqref{eq:sparsity} is therefore a necessary consequence of
assuming that $d_t(\C)\ge2r+3$, as in \eqref{eq:failure}.

\subsection{Upper bounds on the length of a putative counterexample}

In this subsection we show that a counterexample to Conjecture~\ref{conj:packing-covering} (if it exists)
cannot be too long, by deriving upper bounds on its length $n$. Comparing them with the
lower bounds \eqref{eq:cover} and \eqref{eq:coverrough} coming from the covering radius
will be one of our main tools for excluding counterexamples.

For a counterexample, the upper bounds on $n$ will all be derived from \eqref{eq:sparsity}.

To obtain an upper bound on $n$, we will project the columns of a parity-check matrix $H$ onto a suitable quotient of $\Fq^\rho$. We first need the following elementary bound for vectors in the quotient space. For $j\ge1$, set $P_j:=(q^j-1)/(q-1)$, the number of one-dimensional subspaces of $\Fq^j$.

\begin{lemma}\label{lem:weighted-line}
Let $h\ge2$ and let $v_1,\ldots,v_N\in\Fq^h$. Suppose that every two-dimensional subspace of $\Fq^h$ contains at most $B$ of the vectors $v_1,\ldots,v_N$, counted with multiplicity. Write $B=(q+1)\alpha+\beta$, where $0\le\beta\le q$, and define
\begin{equation}\label{eq:U}
U_q(h,B):=
\alpha P_h+
\begin{cases}
0, & \beta=0,\\
1+(\beta-1)P_{h-1}, & \beta\in[q].
\end{cases}
\end{equation}
Then $N\le U_q(h,B)$.

For $q=2$, this bound is sharp for every $h$ and $B$. More precisely,
\begin{equation}\label{eq:Ubinary}
U_2(h,3\alpha+\beta)
=
\alpha(2^h-1)+
\begin{cases}
0, & \beta=0,\\
1, & \beta=1,\\
2^{h-1}, & \beta=2.
\end{cases}
\end{equation}
\end{lemma}

\begin{proof}
Assume first that none of the vectors $v_1,\ldots,v_N$ is zero. For every one-dimensional subspace $J\subseteq\Fq^h$, let
$n_J:=|\{i\in[N]:v_i\in J\}|$, and let $w:=\max_J n_J$. Since there are $P_h$ one-dimensional subspaces of $\Fq^h$, we have $N\le wP_h$.

Choose a one-dimensional subspace $J$ with $n_J=w$. There are $P_{h-1}$ two-dimensional subspaces containing $J$, and every nonzero vector outside $J$ belongs to exactly one of them. Since each such subspace contains at most $B$ of the vectors $v_1,\ldots,v_N$, we obtain
$N\le w+(B-w)P_{h-1}$.

If $w\le\alpha$, then $N\le\alpha P_h$. Suppose instead that $w\ge\alpha+1$. Since $w+(B-w)P_{h-1}$ is nonincreasing in $w$, we get
$$
N\le \alpha+1+(B-\alpha-1)P_{h-1}.
$$
Using $B=(q+1)\alpha+\beta$ and $P_h=qP_{h-1}+1$, this gives
$$
N\le\alpha P_h+1+(\beta-1)P_{h-1}.
$$
If $\beta=0$, the right-hand side is at most $\alpha P_h$. This proves the first statement when none of the $v_i$ is zero.

Now suppose that exactly $z$ of the vectors $v_i$ are zero. Every two-dimensional subspace contains these $z$ vectors, so the remaining nonzero vectors satisfy the same hypothesis with $B$ replaced by $B-z$. Hence,
$N\le z+U_q(h,B-z)$. By \eqref{eq:U}, increasing $B$ by one increases $U_q(h,B)$ by $1$ when $\beta=0$, and by $P_{h-1}$ in every other case, including when $\beta=q$. Thus, $U_q(h,B+1)-U_q(h,B)\ge1$, so $z+U_q(h,B-z)\le U_q(h,B)$.

Finally, let $q=2$. Taking $\alpha$ copies of each nonzero vector attains the bound when $\beta=0$. If $\beta=1$, add one further copy of any nonzero vector. If $\beta=2$, add one further copy of every vector in an affine hyperplane not containing the origin. Every two-dimensional subspace meets such an affine hyperplane in either zero or two vectors, so these constructions attain $U_2(h,B)$ in all three cases.
\end{proof}

We can now apply Lemma~\ref{lem:weighted-line} to obtain an upper bound on the length of a counterexample.

\begin{proposition}\label{prop:weighted-length}
Let $\C$ be an $[n,k]_q$ linear code with redundancy $\rho$, let $t\in[\min\{k,\rho\}]$, let $r=R_t(\C)$, and suppose that $d_t(\C)\ge2r+3$. Set $m:=2r-t+2$ and $\gamma:=\rho-2r+t-2$. Then
\begin{equation}\label{eq:weighted-length}
n\le m-2+U_q(\gamma+2,t+1)
=
2r-t+U_q(\rho-2r+t,t+1).
\end{equation}
In particular, if $q=2$ and $t=3$, then
\begin{equation}\label{eq:binary-length}
n\le2r-3+2^{\rho-2r+3}.
\end{equation}
\end{proposition}

\begin{proof}
Choose $m-2=2r-t$ independent columns of $H$, and let $X\subseteq\Fq^\rho$ be their span. Remove these $m-2$ columns and project all the remaining columns onto $\Fq^\rho/X$. By \eqref{eq:gamma}, the quotient has dimension $\rho-m+2=\gamma+2$.
Let $W\subseteq\Fq^\rho/X$ be a two-dimensional subspace. Its inverse image in $\Fq^\rho$ is an $m$-dimensional subspace containing $X$. By \eqref{eq:sparsity}, it contains at most $m+t-1=2r+1$ columns of $H$, counted with multiplicity. Since the $m-2$ columns spanning $X$ are among them, at most $t+1$ of the remaining columns project into $W$. This includes the columns whose projection is zero.
Therefore, the projected columns satisfy the assumptions of Lemma~\ref{lem:weighted-line} with $h=\gamma+2$ and $B=t+1$. By \eqref{eq:U},
$n-(m-2)\le U_q(\gamma+2,t+1)$, which gives \eqref{eq:weighted-length}.
If $q=2$ and $t=3$, then \eqref{eq:Ubinary} gives $U_2(h,4)=2^h$. Substituting $h=\gamma+2=\rho-2r+t$ and $t=3$ into \eqref{eq:weighted-length} gives \eqref{eq:binary-length}.
\end{proof}

A second upper bound on $n$ can be obtained by averaging over the $m$-dimensional subspaces containing the span of a fixed set of independent columns.

\begin{lemma}\label{lem:incidence}
Let $\C$ be an $[n,k]_q$ linear code with redundancy $\rho$, let $t\in[\min\{k,\rho\}]$, let $r=R_t(\C)$, and suppose that $d_t(\C)\ge2r+3$. Set $m:=2r-t+2$. Then, for every integer $0\le x<m$,
\begin{equation}\label{eq:incidence}
n\le
x+
\left\lfloor
(2r+1-x)
\frac{q^{\rho-x}-1}{q^{m-x}-1}
\right\rfloor.
\end{equation}
\end{lemma}

\begin{proof}
Choose $x$ independent columns of $H$, and let $X\subseteq\Fq^\rho$ be their span. Consider the $m$-dimensional subspaces of $\Fq^\rho$ containing $X$. Each of the chosen $x$ columns belongs to every such subspace, while a column not contained in $X$ belongs to a proportion
$$
\frac{q^{m-x}-1}{q^{\rho-x}-1}
$$
of them. Therefore, the average number of columns of $H$ contained in an $m$-dimensional subspace containing $X$ is at least
$$
x+(n-x)\frac{q^{m-x}-1}{q^{\rho-x}-1}.
$$
On the other hand, by \eqref{eq:sparsity}, every $m$-dimensional subspace contains at most $m+t-1=2r+1$ columns of $H$. Hence,
$$
x+(n-x)\frac{q^{m-x}-1}{q^{\rho-x}-1}\le2r+1.
$$
Rearranging and using that $n$ is an integer gives \eqref{eq:incidence}.
\end{proof}

The bound in Lemma~\ref{lem:incidence}, minimized over $x$, is closely related to the dimension bound in \cite[Lemma~4.1]{EZ}. Here we will use it together with Proposition~\ref{prop:weighted-length}.

Set $m:=2r-t+2$ and define
\begin{equation}\label{eq:M}
M(q,\rho,t,r):=
\min\left\{
\min_{0\le x<m}
\left(
x+
\left\lfloor
(2r+1-x)
\frac{q^{\rho-x}-1}{q^{m-x}-1}
\right\rfloor
\right),
\,
m-2+U_q(\rho-m+2,t+1)
\right\}.
\end{equation}
By Proposition~\ref{prop:weighted-length} and Lemma~\ref{lem:incidence}, every counterexample to Conjecture~\ref{conj:packing-covering} satisfies
\begin{equation}\label{eq:length-upper}
n\le M(q,\rho,t,r).
\end{equation}

In the following, we will repeatedly compare the upper bound \eqref{eq:length-upper} with the lower bounds \eqref{eq:cover} and \eqref{eq:coverrough} coming from the generalized covering radius.

\subsection{Hamming-ball bounds and puncturing}\label{subsec:puncturing}

The bounds obtained in previous subsections restrict the length of a possible counterexample. We now use Hamming balls to obtain upper bounds on generalized Hamming weights of such a counterexample. The main idea is to find a subcode with small support, puncture its support, and repeat the argument on the resulting code. This idea is inspired by the proof of the Griesmer bound; see \cite{Griesmer} and \cite[Section 2.7]{HuffmanPless}.

We first recall the following result from~\cite{YS}. We state it for a code of length $N$
and redundancy $h$, since it will be applied to punctured codes and at test lengths.

\begin{lemma}[{\cite[Lemma~13]{YS}}]\label{lem:hamming-ball}
Let $\D$ be an $[N,N-h]_q$ linear code. If $s\in [N-h]$ and $\ell\ge0$ satisfy
$\vol_q(N,\ell)>q^{h+s-1}$,
then $d_s(\D)\le(s+1)\ell$.
\end{lemma}

\begin{proof}
We recall the argument of \cite[Lemma~13]{YS}, which finds $s+1$ affinely independent
vectors with the same syndrome (where the syndrome of a vector
$x\in\Fq^N$ is $Hx^\top$ for a parity-check matrix $H$ of $\D$;
two vectors have the same syndrome precisely when their difference
belongs to $\D$) inside a Hamming ball of radius $\ell$.
Indeed, since $\vol_q(N,\ell)>q^{h}\cdot q^{s-1}$, some syndrome is attained
by more than $q^{s-1}$ vectors of the Hamming ball of radius $\ell$ centred at $0$. These vectors lie in a coset of $\D$, and an affine subspace of dimension at most $s-1$ has at most $q^{s-1}$ points, so among them there
are $s+1$ affinely independent vectors $x_0,\ldots,x_s$. The differences $x_i-x_0$ span an $s$-dimensional subcode of $\D$ supported on $\bigcup_i\supp(x_i)$, a set of size at most $(s+1)\ell$.
\end{proof}

To apply Lemma~\ref{lem:hamming-ball} repeatedly, we first describe what happens after puncturing the support of a minimum-support subcode.

\begin{lemma}\label{lem:puncture}
Let $\C$ be an $[n,k]_q$ linear code with redundancy $\rho$, let $s\in[k]$, and let $F$ be the support of an $s$-dimensional subcode of $\C$ such that $|F|=d_s(\C)$. Puncturing $\C$ on the coordinates in $F$ gives a code $\D$ of length $n-d_s(\C)$, dimension $k-s$, and redundancy $\rho-d_s(\C)+s$. If $s<k$, then for every $u\in[k-s]$,
$$d_{s+u}(\C)\le d_s(\C)+d_u(\D).$$
\end{lemma}

\begin{proof}
Set $j:=d_s(\C)$. The kernel of the puncturing map consists of the codewords of $\C$ supported on $F$, and hence contains an $s$-dimensional subcode. We claim that its dimension is exactly $s$. This is immediate if $s=k$. Otherwise, if its dimension were at least $s+1$, imposing zero in any coordinate of $F$ would decrease the dimension by at most one and leave an $s$-dimensional subcode supported on at most $j-1$ coordinates, contradicting the definition of $d_s(\C)$.
Thus, the punctured code has dimension $k-s$ and redundancy $(n-j)-(k-s)=\rho-j+s$. If $s<k$ and $\D'$ is a $u$-dimensional subcode of $\D$ with support of size $d_u(\D)$, its inverse image under puncturing has dimension $s+u$ and is supported on $F$ together with the support of $\D'$. This gives the desired inequality.
\end{proof}

In terms of the parity-check matrix $H$, if $F$ is the support of an $s$-dimensional subcode attaining $d_s(\C)$, then $\eta(F)=s$ and $r_H(F)=|F|-s$.

We will also need to enlarge such a set of coordinates without changing its nullity.

\begin{lemma}\label{lem:extension}
Let $F\subseteq[n]$ satisfy $|F|=j$ and $\eta(F)=s$. If $j-s\le x\le\rho$, then $F$ can be enlarged to a set $F'\subseteq[n]$ such that $|F'|=x+s$, $r_H(F')=x$, and $\eta(F')=s$.
\end{lemma}

\begin{proof}
Since $\eta(F)=s$, we have $r_H(F)=j-s$. As long as the span of the columns indexed by the chosen set has dimension smaller than $x$, add a column of $H$ outside this span. Such a column exists because the columns of $H$ span $\Fq^\rho$. Each added column increases both the size and the rank by one, so the nullity remains equal to $s$.
\end{proof}

Thus, when Lemma~\ref{lem:hamming-ball} gives a subcode whose support has size at most a prescribed value, Lemma~\ref{lem:extension} allows us to enlarge its support to that value while keeping the same nullity.

We will also use the following shortening argument.

\begin{lemma}\label{lem:shorten}
Let $\C$ be an $[n,n-\rho]_q$ linear code, let $t\ge1$, and suppose that $n\ge N\ge\rho+t$. Then there exists an $[N,N-\rho]_q$ linear code $\C_0$, obtained by shortening $\C$, such that $d_t(\C)\le d_t(\C_0)$.
Consequently, any upper bound on $d_t$ that holds for every $[N,N-\rho]_q$ linear code also holds for $\C$.
\end{lemma}

\begin{proof}
The coordinate functionals on $\C$ span the dual space of $\C$. Since $N\ge\rho+t$, we have $n-N\le k-t$. We may therefore choose $n-N$ linearly independent coordinate functionals and take their common kernel. This is a subcode of $\C$ of dimension $k-(n-N)=N-\rho$ whose codewords are zero on the corresponding $n-N$ coordinates. Deleting these coordinates gives an $[N,N-\rho]_q$ code $\C_0$.
Every $t$-dimensional subcode of $\C_0$ corresponds to a $t$-dimensional subcode of $\C$ with the same support size. Hence, $d_t(\C)\le d_t(\C_0)$.
\end{proof}

Notice that no information about the generalized covering radius of $\C_0$ is needed. Lemma~\ref{lem:shorten} will only be used to transfer upper bounds on $d_t$ from a fixed length $N$ to the original code.

We can now apply Lemma~\ref{lem:hamming-ball} successively, puncturing the
support obtained at each step before applying the lemma again to the remaining code.

\begin{theorem}\label{thm:successive}
Let $\C$ be an $[n,n-\rho]_q$ linear code, let $t\in[n-\rho]$, and let $B\ge t$.
Let $v\ge1$, and let $(\ell_i,s_i)$, $i\in[v]$, be pairs of positive integers.
Set
$$
S_i:=\sum_{j=1}^i s_j,
\qquad
W_i:=\sum_{j=1}^i \ell_j(s_j+1),
\qquad
S_0=W_0=0.
$$
Suppose that $S_v=t$ and $W_v\le B$. Then $d_t(\C)\le B$ in either of the
following cases:
\begin{enumerate}[label=(\arabic*)]
    \item $n\le B$ or $\rho+t\le B$;

    \item $n>B$, $\rho+t>B$, and, for every $i\in[v]$,
\begin{equation}\label{eq:successive}
    \vol_q(n-W_{i-1},\ell_i)
    >q^{\rho-W_{i-1}+S_{i-1}+s_i-1}.
    \end{equation}
\end{enumerate}
Moreover, in case~{\rm (2)}, it is sufficient to verify the stronger conditions
\begin{equation}\label{eq:successive-uniform}
\vol_q(n-B,\ell_i)
>q^{\rho-W_{i-1}+S_{i-1}+s_i-1}
\end{equation}
for every $i\in[v]$.
\end{theorem}

\begin{proof}
If $n\le B$, then $d_t(\C)\le n\le B$. If $\rho+t\le B$, the generalized Singleton bound \cite[Theorem~1]{Wei} gives $d_t(\C)\le\rho+t\le B$. We may therefore assume $n>B$ and $\rho+t>B$.

After the first $i-1$ steps, assume that $W_{i-1}$ coordinates have been punctured and that the kernel of the resulting puncturing map has dimension $S_{i-1}$. The remaining code $\D_{i-1}$ has length $n-W_{i-1}$, dimension $n-\rho-S_{i-1}$, and redundancy $\rho-W_{i-1}+S_{i-1}$.

Since $s_i\le t-S_{i-1}\le n-\rho-S_{i-1}$, \eqref{eq:successive} and Lemma~\ref{lem:hamming-ball} give $d_{s_i}(\D_{i-1})\le \ell_i(s_i+1)$. Let $F_i$ be the support of an $s_i$-dimensional subcode of $\D_{i-1}$ attaining $d_{s_i}(\D_{i-1})$. By Lemma~\ref{lem:puncture}, this support has nullity exactly $s_i$.

If $|F_i|<\ell_i(s_i+1)$, we enlarge it using Lemma~\ref{lem:extension}, applied to $\D_{i-1}$ with $x:=\ell_i(s_i+1)-s_i$. This is possible because $x$ is at most the redundancy $\rho-W_{i-1}+S_{i-1}$ of $\D_{i-1}$; indeed, $W_i-S_i$ is increasing in $i$, so
$W_i-S_i\le W_v-t\le B-t<\rho$.
We may therefore puncture exactly $\ell_i(s_i+1)$ additional coordinates while increasing the dimension of the kernel by exactly $s_i$.

After the last step, the kernel of the composite puncturing map has dimension $S_v=t$ and is supported on $W_v\le B$ coordinates. Hence, $d_t(\C)\le B$.

Finally, since $W_{i-1}\le B$ and $\vol_q(N,\ell)$ is nondecreasing in $N$,
\eqref{eq:successive-uniform} implies \eqref{eq:successive}.
\end{proof}

The following consequence, obtained by taking $\ell_i=2$ at every step, will be useful later.

\begin{corollary}\label{cor:radius-two}
Let $\C$ be an $[n,n-\rho]_q$ linear code, let $t\in[n-\rho]$, and let $L,A\ge1$. If
$t\le A(2^L-1)$, $n>2t+2L$, and
$\vol_q(n-2t-2L,2)>q^{\rho+A-1}$, then
$d_t(\C)\le2t+2L$.
\end{corollary}

\begin{proof}
If $\rho+t\le2t+2L$, the result follows from the generalized Singleton bound. We may therefore assume $\rho+t>2t+2L$.

Set $S_0=0$ and, as long as $S_{i-1}<t$, let
$s_i:=\min\{t-S_{i-1},A+S_{i-1}+2(i-1)\}$ and $S_i:=S_{i-1}+s_i$.
Before reaching $t$, this gives $S_i=2S_{i-1}+A+2(i-1)$, and hence
$S_i=A(2^i-1)+2(2^i-i-1)$. Since $t\le A(2^L-1)$, the process reaches $S_i=t$ for some $i\le L$.
Take $\ell_i=2$ at every step. Then $W_i=2(S_i+i)\le2t+2L$. Moreover,
$\rho-W_{i-1}+S_{i-1}+s_i-1
=\rho-S_{i-1}-2(i-1)+s_i-1
\le\rho+A-1$.
Thus, with $B=2t+2L$, the assumption
$\vol_q(n-2t-2L,2)>q^{\rho+A-1}$
implies \eqref{eq:successive-uniform} at every step. Therefore, Theorem~\ref{thm:successive}, applied with $B=2t+2L$, gives
$d_t(\C)\le2t+2L$.
\end{proof}

We have thus obtained the main bounds that will be used in the rest of the paper.
For a counterexample, \eqref{eq:length-upper} gives an upper bound on $n$,
while \eqref{eq:cover} and \eqref{eq:coverrough} give lower bounds.
These bounds will be compared in the next section. When they do not yield a
contradiction, Theorem~\ref{thm:successive} and
Corollary~\ref{cor:radius-two} will be used to prove
$d_t(\C)\le2r+2$.

\section{Reduction to finitely many parameters}\label{sec:reduction}

We now use the bounds from Sections~\ref{sec:prelim} and~\ref{sec:putative_code}
to reduce the proof of Conjecture~\ref{conj:packing-covering} to finitely
many parameters. We first combine the known cases with the bounds on the
length of a putative counterexample to settle small redundancy. We then show that
a putative counterexample cannot have large redundancy when $t$ and $r$ are fixed,
or a small gap $r-t$ when $t$ is large. Finally, we settle all orders
$t\ge32$ and bound $r$ for each of the remaining orders. This leaves
finitely many tuples $(q,\rho,t,r)$, which we exclude by exact computation.

Note that the arguments in this section are more technical than those in the
preceding sections. Their purpose is to rule out entire ranges of
parameters, until only a finite set remains. We give the inequalities
used in each reduction and describe the finite calculations in the
appendix.

\subsection{Initial reductions}

We first collect the immediate restrictions that the parameters of a putative counterexample must satisfy.
By Proposition~\ref{prop:inputs}, \eqref{eq:failure-range}, and
\eqref{eq:rate}, every counterexample to
Conjecture~\ref{conj:packing-covering} satisfies
\begin{equation}\label{eq:finite-parameters}
\begin{gathered}
3\le t\le\rho-5,\qquad
t+1\le r\le\left\lfloor\frac{\rho+t-3}{2}\right\rfloor,
\qquad 3r<2\rho,\\
q<r,\qquad q\text{ a prime power},\qquad
(t,r)\notin\{(3,4),(3,5),(4,5),(4,6)\},\\
n\ge \rho+5t-1,\qquad
n\ge\left\lfloor\frac{5\rho}{2}\right\rfloor+1.
\end{gathered}
\end{equation}

The restrictions on $t$, $r$, and $q$ follow from
Proposition~\ref{prop:inputs} and~\eqref{eq:failure-range}. The lower bounds on $n$ follow from~\eqref{eq:rate} and $n=k+\rho$.

For fixed $q,\rho,t,r$, the length $n$ is not a priori bounded from below by
a single explicit integer sharp enough for every argument. We therefore use
a test length $N_0\ge\rho+t$. Such a length is admissible whenever one of the
lower bounds in \eqref{eq:finite-parameters} gives $n\ge N_0$, or whenever
the generalized ball-covering bound gives
$$
\vol_{q^t}(N_0-1,r)<q^{t\rho},
$$
since in this case \eqref{eq:cover} implies $n\ge N_0$. Faster sufficient
tests derived from \eqref{eq:coverrough} are also used; they are listed in
Appendix~\ref{app:test-length}.

At such a test length, there are two ways to exclude a counterexample.
 If
$$
N_0>M(q,\rho,t,r),
$$
then the lower bound $n\ge N_0$ contradicts \eqref{eq:length-upper}.
Otherwise, it is enough to prove
$d_t(\C_0)\le2r+2$ for every $[N_0,N_0-\rho]_q$ code $\C_0$ and then apply
Lemma~\ref{lem:shorten}, which gives the same bound for every possible original
length $n\ge N_0$.

We first use these bounds to settle all codes with redundancy at most $50$.
The remaining parameter tuples are checked by the program described in
Appendix~\ref{app:computations}.

\begin{theorem}\label{thm:finite}
Let $\C$ be an $[n,k]_q$ linear code with redundancy $\rho\le50$,
and let $t\in[\min\{k,\rho\}]$. Then
$$
d_t(\C)\le2R_t(\C)+2.
$$
\end{theorem}

\begin{proof}
Set $r:=R_t(\C)$. If the parameters do not satisfy
\eqref{eq:finite-parameters}, the result follows from
Proposition~\ref{prop:inputs} or the generalized Singleton bound.
We complete the proof by an exact computer calculation over the
remaining parameter tuples $(q,\rho,t,r)$ with $\rho\le50$. For each tuple, we establish a lower bound $n\ge N_0$
as described above. We then check one of two inequalities. If
$N_0>M(q,\rho,t,r)$, the tuple is excluded by
\eqref{eq:length-upper}. Otherwise, the generalized-weight bounds of
Appendix~\ref{app:weights}, together with Theorem~\ref{thm:successive},
show that every $[N_0,N_0-\rho]_q$ code satisfies $d_t\le2r+2$.
Lemma~\ref{lem:shorten} then gives the same bound for $\C$.
These exact calculations exclude all $59{,}086$ admissible tuples:
$39{,}627$ by the length comparison and $19{,}459$ by generalized-weight
bounds. No codes are enumerated.
\end{proof}

\subsection{Large redundancy and small radius gaps}

We now prove two results that will reduce the parameters of a putative
counterexample. Set
$$a:=r-t.$$
First, we show that if $a$ is fixed,
then $t$ cannot be too large. The proof compares the upper and lower
bounds on the length obtained in the previous section. This gives a
lower bound on $\rho/r$ that allows us to apply
Corollary~\ref{cor:radius-two}.

We then show that, for each fixed pair $(t,r)$, the redundancy $\rho$
of a putative counterexample is bounded. Together, these results leave only
finitely many parameters for each fixed value of $a$. We will use
this to settle the first five positive values of $a$.

\begin{theorem}\label{thm:diagonal}
Let $\C$ be an $[n,k]_q$ linear code with redundancy $\rho$, let $t\in[\min\{k,\rho\}]$, and set $r:=R_t(\C)$ and $a:=r-t\ge1$. If
$t\ge\max\{132,5a^2\}$, then $d_t(\C)\le2r+2$.
\end{theorem}

\begin{proof}
Suppose that $d_t(\C)\ge2r+3$. We compare the upper bound on $n$
from Lemma~\ref{lem:incidence} with the lower bound
\eqref{eq:coverrough}. This will show that $\rho/r$ is large enough
to apply Corollary~\ref{cor:radius-two}.

By Proposition~\ref{prop:inputs}, we may assume $3r<2\rho$. Set
$m:=2r-t+2=t+2a+2$ and $\lambda:=\rho/r$. The assumption
$t\ge5a^2$ gives
$$
\frac at\le\frac1{\sqrt{5t}}\le\frac1{25}.
$$
Since $t\ge132$, it follows that $m<2t$ and
$$
\rho-m>\frac t2-\frac a2-2>\frac t3.
$$

We first bound $n$ from above. Let $b$ be the least positive integer
such that $q^b\ge t$. This choice makes the bound in
Lemma~\ref{lem:incidence} particularly simple. We have
\begin{equation}\label{eq:b-log}
b+1\le\frac t{14}.
\end{equation}
Indeed, if $2^u\le t<2^{u+1}$, then $u\ge7$ and $b+1\le u+2$.
The inequality $14(u+2)\le2^u$ holds for $u=7$ and remains true
as $u$ increases.

Apply Lemma~\ref{lem:incidence} with parameter $m-b$. Since
$q^b\ge t$ and $b\le t-1$, we obtain
$$
n\le m+\frac{t+b-1}{1-q^{-b}}q^{\rho-m} \le m+\frac{t+b-1}{1-\frac1t}q^{\rho-m}=m+\frac{t}{t-1}(t+b-1)q^{\rho-m}
\le m+(t+b+1)q^{\rho-m}.
$$
Moreover, $m<2t$ and $\rho-m>t/3$ give
$mq^{m-\rho}<2t\,2^{-t/3}<t/84$. Together with
\eqref{eq:b-log}, this yields
\begin{equation}\label{eq:diagonal-upper}
n\le\frac{13}{12}tq^{\rho-m}.
\end{equation}

We now compare this upper bound with the lower bound
\eqref{eq:coverrough}, which gives
$$
n>\frac r{\mathrm e}q^{t(\lambda-1)}.
$$
Suppose that $\lambda\le2+2/(5a)$. Since $\rho=r\lambda=(t+a)\lambda$,
we have
$$
t(\lambda-1)-(\rho-m)=2+a(2-\lambda)\ge\frac85.
$$
Using $r\ge t$, $q\ge2$, $2^{8/5}>3$, and
$\mathrm e<11/4$, we find
$$
n>\frac{t}{\mathrm e}q^{\rho-m}2^{8/5}
>\frac{12}{11}tq^{\rho-m}
>\frac{13}{12}tq^{\rho-m}.
$$
This contradicts \eqref{eq:diagonal-upper}. Hence,
\begin{equation}\label{eq:x-threshold}
\lambda>2+\frac{2}{5a}.
\end{equation}

We use this inequality to check the conditions of
Corollary~\ref{cor:radius-two}. Set
$L:=a+1$ and $B:=2t+2L=2r+2$. Thus, the bound given by
the corollary is exactly the bound we need. Take
$$
A:=\left\lceil\frac{t}{2^{a+1}-1}\right\rceil,
$$
the smallest integer for which $t\le A(2^L-1)$. In particular,
$A\le t/(2^{a+1}-1)+1$.

Since $\lambda>2$, \eqref{eq:coverrough} gives
$n>(r/\mathrm e)q^{t(\lambda-1)}>(r/3)2^t>2B$.
Therefore, $N:=n-B>n/2$. Using $q-1\ge q/2$ and
\eqref{eq:coverrough} again, we obtain
$$
\vol_q(N,2)
\ge\binom N2(q-1)^2
>\frac{n^2q^2}{64}
>\frac{r^2}{576}q^{2t(\lambda-1)+2}
\ge q^{2t(\lambda-1)+2},
$$
where the last inequality holds because $r\ge t\ge132$.

It remains to show that $2t(\lambda-1)+2>\rho+A-1$. By
\eqref{eq:x-threshold}, $t-a>0$, and the bound on $A$,
\begin{align*}
2t(\lambda-1)+2-(\rho+A-1)
&=(t-a)\lambda-2t+3-A\\
&>
t\left(\frac{2}{5a}-\frac{1}{2^{a+1}-1}\right)
-2a+\frac85.
\end{align*}
The coefficient of $t$ is positive, since $5a<2^{a+2}-2$ for every $a\ge1$. For $1\le a\le5$, the
right-hand side is positive at $t=132$: its values are,
respectively, $42/5$, $36/7$, $22/5$, $394/155$, and $34/525$.
It is therefore positive for every $t\ge132$. For $a\ge6$,
the assumption $t\ge5a^2$ gives the lower bound
$$
\frac85-\frac{5a^2}{2^{a+1}-1}>0.
$$
The last inequality holds at $a=6$, and
$a^2/(2^{a+1}-1)$ decreases thereafter.

Thus, $\vol_q(n-B,2)>q^{\rho+A-1}$. All the conditions of
Corollary~\ref{cor:radius-two} hold, so
$d_t(\C)\le B=2r+2$, a contradiction.
\end{proof}

Theorem~\ref{thm:diagonal} bounds $t$ when the gap $r-t$ is fixed. We now bound the redundancy $\rho$ of a counterexample for every fixed pair $(t,r)$.
More precisely, we give an explicit threshold above which a single
Hamming ball yields the desired bound on $d_t(\C)$.

\begin{theorem}\label{thm:tail}
Let $\C$ be an $[n,k]_q$ linear code with redundancy $\rho$, let
$3\le t\le\min\{k,\rho\}$, and set $r:=R_t(\C)$. Define
$$
\ell:=\left\lfloor\frac{2r+2}{t+1}\right\rfloor
$$
and
$$
D(t,r):=
\left\lceil
\frac{r\bigl((\ell+1)t+\ell-1\bigr)}{\ell t-r}
\right\rceil.
$$
Then $\ell t>r$. If $\rho\ge D(t,r)$, then
$d_t(\C)\le2r+2$.
\end{theorem}

\begin{proof}
We want to apply Lemma~\ref{lem:hamming-ball} with $s=t$.
The choice of $\ell$ gives $(t+1)\ell\le2r+2$, so the conclusion
of that lemma will be the bound we need.

We first check that $\ell t>r$. Write $r=ut+v$, where
$u\ge1$ and $0\le v<t$. The integer $\ell_0:=u+1$ satisfies
$$
2r+2-\ell_0(t+1)=(u-1)(t-1)+2v\ge0.
$$
Hence, $\ell\ge \ell_0>r/t$, as claimed. We also have
$$
\ell\le\frac{2(r+1)}{t+1}\le\frac{2r}{t}
$$
because $r\ge t$. In particular, $\ell\le r$ and
$r/(6\ell)\ge t/12\ge1/4$.

We now bound the volume of a Hamming ball of radius $\ell$.
Since $t\le k$ and $r\le\rho$, we have
$n\ge\rho+t\ge r+t>\ell$. Also,
$$
\binom{n}{\ell}
=\prod_{i=0}^{\ell-1}\frac{n-i}{\ell-i}
\ge\left(\frac{n}{\ell}\right)^\ell:
$$
each factor is at least $n/\ell$ because $n\ge \ell$.
Using $q-1\ge q/2$, $\mathrm e<3$, and
\eqref{eq:coverrough}, we obtain
\begin{align*}
\vol_q(n,\ell)
&\ge\binom{n}{\ell}(q-1)^\ell\ge\left(\frac{nq}{2\ell}\right)^\ell>\left(\frac{r}{6\ell}\right)^\ell
q^{\ell t(\rho-r)/r+\ell}
\ge q^{\ell t(\rho-r)/r-\ell}.
\end{align*}
For the last inequality, we used
$(r/(6\ell))^\ell\ge4^{-\ell}\ge q^{-2\ell}$.

To apply Lemma~\ref{lem:hamming-ball} with $s=t$, we need
$\vol_q(n,\ell)>q^{\rho+t-1}$. The preceding bound gives this
whenever
$$
\frac{\ell t(\rho-r)}{r}-\ell\ge\rho+t-1,
$$
or, equivalently,
$$
\frac{\ell t-r}{r}\rho\ge(\ell+1)t+\ell-1.
$$
Since $\ell t>r$, the last inequality holds when
$\rho\ge D(t,r)$. Lemma~\ref{lem:hamming-ball} therefore gives
$$
d_t(\C)\le(t+1)\ell\le2r+2,
$$
as required.
\end{proof}

Theorems~\ref{thm:diagonal} and~\ref{thm:tail} leave only finitely many
parameters for each fixed positive value of $a=r-t$. We now use this
to settle $1\le a\le5$.

\begin{theorem}\label{thm:fivegaps}
Let $\C$ be an $[n,k]_q$ linear code with redundancy $\rho$,
and let $t\in[\min\{k,\rho\}]$. If $R_t(\C)\le t+5$, then $d_t(\C)\le2R_t(\C)+2$.
\end{theorem}

\begin{proof}
Set $r:=R_t(\C)$. The case $r=t$ and the orders $t\le2$ are
covered by Proposition~\ref{prop:inputs}. We may therefore assume
$t\ge3$ and set $a:=r-t\in\{1,\ldots,5\}$.

If $t\ge132$, the result follows from
Theorem~\ref{thm:diagonal}, since $5a^2\le125$. If $\rho\le50$,
it follows from Theorem~\ref{thm:finite}, while if
$\rho\ge D(t,r)$ it follows from Theorem~\ref{thm:tail}.
Finally, the case $q\ge r$ is covered by
Proposition~\ref{prop:inputs}.

Consequently, a counterexample must satisfy
\begin{equation}\label{eq:gap-range}
\begin{gathered}
1\le a\le5,\qquad 3\le t\le131,\qquad r=t+a,\qquad
q<r\text{ a prime power},\\
\max\left\{51,\,2r-t+3,\,
\left\lfloor\frac{3r}{2}\right\rfloor+1\right\}
\le\rho<D(t,r).
\end{gathered}
\end{equation}
Here $\rho\ge2r-t+3$ follows from
\eqref{eq:failure-range}, while
$\rho\ge\lfloor3r/2\rfloor+1$ follows from $3r<2\rho$.
The four pairs excluded in \eqref{eq:finite-parameters} do not occur in \eqref{eq:gap-range}, since $D(t,r)\le39$ for those pairs and the latter range requires $\rho\ge51$. The other restrictions on $(q,\rho,t,r)$ in \eqref{eq:finite-parameters} follow from \eqref{eq:gap-range}.

We complete the proof by an exact calculation over the
$2{,}396{,}360$ tuples $(q,\rho,t,r)$ in
\eqref{eq:gap-range}. For each tuple, the program computes the
lower bound $n\ge N_0$ of Appendix~\ref{app:test-length}, which
follows from \eqref{eq:finite-parameters} and \eqref{eq:coverrough}.
In $1{,}353{,}595$ cases $N_0>M(q,\rho,t,r)$, and the tuple is
excluded by \eqref{eq:length-upper}. In $1{,}033{,}256$ cases the
conditions of Lemma~\ref{lem:app-fixed-radius} hold at length $N_0$,
with $\ell=2$, $L=a+1$ and $A$ as in Appendix~\ref{app:test-length}, so
every $[N_0,N_0-\rho]_q$ code satisfies $d_t\le2r+2$. For the remaining
$9{,}509$ tuples the input records supply a sharper test length $N$,
justified as in Appendix~\ref{app:test-length}. In $8{,}848$ of these
cases $N>M(q,\rho,t,r)$. In the other $661$ cases, the
generalized-weight bounds of Appendix~\ref{app:weights} show that every
$[N,N-\rho]_q$ code satisfies $d_t\le2r+2$. Whenever the bound has been
proved at a test length, Lemma~\ref{lem:shorten} gives the same bound
for $\C$. Thus, no tuple in \eqref{eq:gap-range} can correspond to a
counterexample. No codes are enumerated.
\end{proof}

We have therefore excluded every counterexample with $r-t\le5$.
It remains to consider gaps $r-t\ge6$. We first show that all such
counterexamples are excluded when the order $t$ is sufficiently large.

\subsection{The case \texorpdfstring{$t\ge32$}{t >= 32}}\label{sec:highorder}

We now prove the conjecture for all orders $t\ge32$.
By Theorem~\ref{thm:fivegaps}, it remains to consider
$a=r-t\ge6$.

Our first step is to show that a counterexample in this
range must satisfy a stronger lower bound on $\rho/r$
than the bound $3r<2\rho$. The stronger bound
gives the extra factor needed to estimate
Hamming balls after puncturing the code. We will then use
these estimates to prove the conjecture separately
in the ranges
$$r\le2t, \qquad 2t\le r\le3t, \qquad r\ge3t.$$

\begin{lemma}\label{lem:ratio-new}
Let $\C$ be an $[n,k]_q$ linear code with redundancy
$\rho$, let $t\in[\min\{k,\rho\}]$, and set
$r:=R_t(\C)$ and $a:=r-t$. If $\C$ is a
counterexample for order $t$, $t\ge32$, and
$a\ge6$, then
\begin{equation}
\frac{\rho}{r}>2+\frac{2}{5a}.
\end{equation}
\end{lemma}

\begin{proof}
Set $m:=2r-t+2$, $\gamma:=\rho-m$, and
$\lambda:=\rho/r$. By \eqref{eq:failure-range},
$\gamma\ge1$. Let $b$ be the least positive
integer such that $q^b\ge t$. Applying
Lemma~\ref{lem:incidence} with parameter $m-b$,
as in the proof of Theorem~\ref{thm:diagonal},
gives
\begin{equation}\label{eq:ratio-cap}
n\le m+(t+b+1)q^\gamma.
\end{equation}
Here $b\le t-1$, so $m-b\ge0$. Moreover,
\begin{equation}\label{eq:ratio-b}
b+1\le\frac t4.
\end{equation}
Indeed, if $2^j\le t<2^{j+1}$, then $j\ge5$ and
$b+1\le j+2\le2^j/4\le t/4$.

We first prove that $\lambda>2$. Suppose that
$\lambda\le2$. Since $\gamma=\rho-2r+t-2$,
the covering bound \eqref{eq:coverrough} gives
\begin{equation}\label{eq:ratio-cover}
nq^{-\gamma}>
\frac r{\mathrm e}q^{2+a(2-\lambda)}.
\end{equation}
If $\gamma\ge4$, then \eqref{eq:ratio-cap},
\eqref{eq:ratio-b}, and $m\le2r$ give
$$
nq^{-\gamma}
\le\frac r8+\frac{5t}{4}
\le\frac{11r}{8}.
$$
On the other hand, \eqref{eq:ratio-cover} gives
$nq^{-\gamma}>4r/\mathrm e>16r/11$,
a contradiction.

If $\gamma\le3$, then
$$
a(2-\lambda)
=\frac{a(t-2-\gamma)}{t+a}
\ge\frac{6(t-5)}{t+6}
\ge1.
$$
The last inequality holds for $t\ge32$.
Consequently \eqref{eq:ratio-cover} gives
$nq^{-\gamma}>8r/\mathrm e>32r/11$,
whereas \eqref{eq:ratio-cap} gives
$$
nq^{-\gamma}
\le\frac m2+\frac{5t}{4}
\le\frac{9r}{4}.
$$
This is again a contradiction. Hence, $\lambda>2$.

Suppose now that
$\lambda\le2+2/(5a)$. Then
$2+a(2-\lambda)\ge8/5$, and
\eqref{eq:ratio-cover} gives
\begin{equation}\label{eq:ratio-lower}
nq^{-\gamma}>\frac{12r}{11},
\end{equation}
since $q^{8/5}\ge2^{8/5}>3$ and
$\mathrm e<11/4$.

As $\lambda>2$, we have $\gamma>t-2$.
Thus, \eqref{eq:ratio-cap} implies
\begin{equation}\label{eq:ratio-upper}
nq^{-\gamma}
\le t+b+1+2^{3-t}r.
\end{equation}
If $t\ge128$, the argument for
\eqref{eq:b-log} gives $b+1\le t/14$.
Since $2^{3-t}<1/128$ and $t\le r$,
\eqref{eq:ratio-upper} yields
$$
nq^{-\gamma}
<
\left(\frac{15}{14}+\frac1{128}\right)r
<\frac{12r}{11},
$$
contrary to \eqref{eq:ratio-lower}.

If $32\le t<128$, then $b+1\le8$.
Since $r\ge t+6$ and
$$
\left(\frac{12}{11}-\frac1{128}\right)(t+6)
-(t+8)
=\frac{117t-2114}{1408}>0,
$$
\eqref{eq:ratio-upper} is again smaller
than $12r/11$. This contradicts
\eqref{eq:ratio-lower} and proves
the statement.
\end{proof}

For the rest of this subsection, set $B:=2r+2$. The next lemma gives a lower bound for a Hamming ball of radius $\ell$ after up to $B$ coordinates of $\C$ have been removed.

\begin{lemma}\label{lem:ball-exponent}
Let $\C$ be an $[n,k]_q$ linear code with redundancy
$\rho$, let $t\in[\min\{k,\rho\}]$, and set
$r:=R_t(\C)$ and $a:=r-t$. Suppose that $\C$
is a counterexample for order $t$, with $t\ge32$
and $a\ge6$. If an integer $\ell\ge1$ satisfies
$\ell t\ge r$ and $12\ell\le r$, then $n>2B$ and
\begin{equation}
\vol_q(n-B,\ell)>q^{\rho+E_\ell},
\qquad
E_\ell:=\ell(t+1)-2r+
\frac{2(\ell t-r)}{5a}.
\end{equation}
In particular,
$\vol_q(n-B,\ell)>q^{\rho+\ell(t+1)-2r}$.
\end{lemma}

\begin{proof}
Set $\lambda:=\rho/r$. By
Lemma~\ref{lem:ratio-new}, $\lambda>2$.
The covering bound \eqref{eq:coverrough} gives
\begin{equation}\label{eq:ball-length}
n>\frac r{\mathrm e}q^{t(\lambda-1)}.
\end{equation}
In particular, $n>2B$. Indeed,
$r\ge t+6\ge38$ and
$q^{t(\lambda-1)}>2^t$, so the right-hand
side of \eqref{eq:ball-length} exceeds
$4r+4=2B$.

Set $N:=n-B$. Then $N>n/2$ and $N>\ell$.
Using $q-1\ge q/2$ and
$\binom N\ell\ge(N/\ell)^\ell$, we obtain
$$
\vol_q(N,\ell)
\ge\binom N\ell(q-1)^\ell
>\left(\frac{nq}{4\ell}\right)^\ell.
$$
By \eqref{eq:ball-length}, the last expression
is greater than
$$
\left(\frac{r}{4\mathrm e\ell}\right)^\ell
q^{\ell(t(\lambda-1)+1)}.
$$
Since $12\ell\le r$ and $\mathrm e<3$, we have
$r/(4\mathrm e\ell)>1$. Therefore,
$\vol_q(n-B,\ell)>
q^{\ell(t(\lambda-1)+1)}$.
Finally,
$$
\ell(t(\lambda-1)+1)-(\rho+E_\ell)
=(\ell t-r)
\left(\lambda-2-\frac{2}{5a}\right)
\ge0
$$
by Lemma~\ref{lem:ratio-new} and
$\ell t\ge r$. This proves the first inequality in the statement. The second one follows from $\ell t\ge r$.
\end{proof}

We now use Lemma~\ref{lem:ball-exponent} to show that a putative counterexample $\C$ contains a $t$-dimensional subcode supported on at most $B=2r+2$ coordinates. When $r\le2t$, we use
a computer calculation for $6\le r-t\le9$
and two puncturing arguments for $r-t\ge10$.

\begin{proposition}\label{prop:large-diagonal}
Let $\C$ be an $[n,k]_q$ linear code with
redundancy $\rho$, let
$32\le t\le\min\{k,\rho\}$, and set
$r:=R_t(\C)$. If $r\le2t$, then
$d_t(\C)\le2r+2$.
\end{proposition}

\begin{proof}
Suppose, for a contradiction, that
$d_t(\C)\ge2r+3$. By
Theorem~\ref{thm:fivegaps}, we may assume
$a:=r-t\ge6$. Since $r\le2t$, we also
have $a\le t$. We distinguish the cases
$6\le a\le9$, $a\ge10$ with $t\ge3a$,
and $a\ge10$ with $t<3a$.

\medskip
\noindent\textbf{Case 1: $6\le a\le9$.}
If $t\ge5a^2$, the result follows from
Theorem~\ref{thm:diagonal}. Thus, assume
$32\le t<5a^2$. There are
$$
\sum_{a=6}^{9}(5a^2-32)=1022
$$
pairs $(t,a)$ in this range.

For each pair $(t,a)$, the input records supply a
sequence of positive integer pairs $(\ell_i,s_i)$,
which the program described in
Appendix~\ref{app:steps} checks. At step $i$, it uses a Hamming
ball of radius $\ell_i$ to increase the
nullity by $s_i$, as in
Theorem~\ref{thm:successive}. In the notation
of that theorem, it checks $S_v=t$, $W_v\le B$,
and, for every $i\in[v]$,
\begin{equation}\label{eq:symbolic-chain}
\ell_i t\ge r,\qquad
12\ell_i\le r,\qquad
s_i-1-W_{i-1}+S_{i-1}
\le \ell_i(t+1)-2r+
\frac{2(\ell_i t-r)}{5a}.
\end{equation}
The first two inequalities allow us to apply
Lemma~\ref{lem:ball-exponent}. Together with
the last inequality, it gives
$$
\vol_q(n-B,\ell_i)>
q^{\rho-W_{i-1}+S_{i-1}+s_i-1}
$$
at every step. Hence, all the conditions of
Theorem~\ref{thm:successive} hold, and
$d_t(\C)\le B$, yielding a contradiction.

\medskip
\noindent\textbf{Case 2: $a\ge10$ and $t\ge3a$.}
We first obtain an $(a-1)$-dimensional subcode
of $\C$ supported on at most $3a$ coordinates.
To this end, we apply Lemma~\ref{lem:ball-exponent}
with $\ell=3$. Its hypotheses are satisfied:
since $t\ge3a$, we have
$3t\ge r=t+a$, and since $a\ge10$, we have
$r=t+a\ge4a\ge40>36=12\cdot3$. Moreover,
$$
E_3
=3(t+1)-2(t+a)+\frac{2(3t-r)}{5a}
\ge t-2a+3
\ge a+3.
$$
Thus,
$\vol_q(n,3)\ge\vol_q(n-B,3)>q^{\rho+a-2}$.
Lemma~\ref{lem:hamming-ball}, applied with
$s=a-1$, gives $d_{a-1}(\C)\le3a$.

Let $F$ be the support of an
$(a-1)$-dimensional subcode attaining
$d_{a-1}(\C)$. By Lemma~\ref{lem:puncture},
$\eta(F)=a-1$. We apply Lemma~\ref{lem:extension} with
$x:=2a+1$. Since $|F|\le3a$ and
$\eta(F)=a-1$, we have
$|F|-\eta(F)\le2a+1$. Also
$2a+1\le\rho$, because $\rho>2r$
by Lemma~\ref{lem:ratio-new}. Thus, there
is a set $F'\supseteq F$ with
$|F'|=3a$ and $\eta(F')=a-1$.
Puncturing $\C$ on $F'$ gives a code
$\C'$ of length $n':=n-3a$ and
redundancy $\rho':=\rho-2a-1$.

The coordinates in $F'$ already give nullity $a-1$.
We therefore need to obtain a further nullity
$u:=t-(a-1)=t-a+1$ in the punctured code
$\C'$. We will apply
Corollary~\ref{cor:radius-two} with
$$
L:=\left\lfloor\frac a2\right\rfloor,
\qquad
A:=4+\left\lfloor\frac{2(t-a)}{5a}\right\rfloor.
$$
We first show that $u\le A(2^L-1)$.
For every $a\ge10$, we have
\begin{equation}\label{eq:highorder-power}
2(2^L-1)\ge5a.
\end{equation}
This holds for $a=10,11$ and remains
true when $a$ increases by two, since
the left-hand side doubles and the
right-hand side increases by $10$.
As $A\ge3+2(t-a)/(5a)$, we obtain
$$
A(2^L-1)
\ge3(2^L-1)+(t-a)
\ge t-a+1=u.
$$

We next check the length hypothesis of
Corollary~\ref{cor:radius-two} for $\C'$.
Since $2L\le a$, we have
$$
3a+2u+2L
=2t+a+2+2L
\le2t+2a+2=B.
$$
As $\C$ is a counterexample, $n>B$.
Therefore, $n'=n-3a>2u+2L$.

It remains to check the volume hypothesis of
Corollary~\ref{cor:radius-two} for $\C'$.
We apply Lemma~\ref{lem:ball-exponent} with
$\ell=2$. Its hypotheses hold because
$r=t+a\le2t$ and $r\ge40>24$.
Writing $X:=2(t-a)/(5a)$, the definition
$A=4+\lfloor X\rfloor$ gives
$$
E_2=2-2a+X
\ge2-2a+\lfloor X\rfloor
=A-2a-2.
$$
Moreover, the inequality
$3a+2u+2L\le B$ proved above gives
$n'-2u-2L\ge n-B$. Since the volume
of a Hamming ball is nondecreasing in
the length, we obtain
$$
\vol_q(n'-2u-2L,2)
\ge\vol_q(n-B,2)
>q^{\rho+A-2a-2}
=q^{\rho'+A-1}.
$$

Corollary~\ref{cor:radius-two} now gives
$d_u(\C')\le2u+2L$. The kernel of the
puncturing map from $\C$ to $\C'$ has
dimension $\eta(F')=a-1$. Hence, the
inverse image of a $u$-dimensional
subcode of $\C'$ has dimension
$(a-1)+u=t$ and is supported on at most
$3a+2u+2L\le B$ coordinates. This
contradicts $d_t(\C)>B$.

\medskip
\noindent\textbf{Case 3: $a\ge10$ and $t<3a$.}
We apply Theorem~\ref{thm:successive} using
Hamming radii $4$, $3$, and $2$. The radius-$4$
steps will give the initial nullity needed
to use radius $3$; the radius-$3$ steps
will bring the nullity to $a$; and the
radius-$2$ steps will bring it to $t$.

We first determine how much nullity can
be obtained at one step. Suppose that
the previous steps have produced nullity
$S$ using $W$ coordinates. A step with
Hamming radius $\ell$ that adds nullity
$s$ requires, by
Theorem~\ref{thm:successive},
$$
\vol_q(n-B,\ell)>
q^{\rho-W+S+s-1}.
$$
Lemma~\ref{lem:ball-exponent} gives this
inequality whenever
\begin{equation}\label{eq:highorder-step}
s\le\ell(t+1)-2r+1+W-S.
\end{equation}

Set
$$
c:=\max\left\{0,
\left\lceil\frac{2a-t}{3}\right\rceil
\right\}.
$$
If $c>0$, we use radius-$4$ steps
until the total nullity reaches $c$.
The inequality $t-2a+3c\ge0$
ensures that a radius-$3$ step is
then possible. If $c=0$, we begin
with radius $3$.

Here is the rule for the steps. We use radius $4$ until the total nullity
is $c$, radius $3$ until it is $a$, and
radius $2$ until it is $t$. Suppose that
the current total nullity is $S$ and the
current radius is $\ell$. Let $d$ be
$c-S$, $a-S$, or $t-S$, according as
$\ell$ is $4$, $3$, or $2$. Thus, $d$
is the additional nullity needed before
we change radius or finish.

As long as $d>0$, choose
\begin{equation}\label{eq:phase-rule}
s:=\min\left\{
d,\,
\ell(t+1)-2r+1+W-S
\right\}.
\end{equation}
After this step, replace $S$ by $S+s$
and $W$ by $W+\ell(s+1)$. The choice
$s\le d$ stops us at the prescribed
nullity, while the other bound on $s$
is exactly \eqref{eq:highorder-step}.

We check that the second quantity in
\eqref{eq:phase-rule} is positive
whenever a step is needed. If $c>0$,
the first step uses radius $4$.
At that point $S=W=0$, and the
second quantity equals
$$
4(t+1)-2r+1=2(t-a)+5\ge5.
$$

Let $L_4$ be the number of radius-$4$
steps, taking $L_4=0$ if $c=0$.
After these steps, $S=c$ and
$W=4(c+L_4)$. Thus, when we begin
using radius $3$, the upper bound
on $s$ in \eqref{eq:highorder-step} is
$$
3(t+1)-2r+1+W-S
=t-2a+4+3c+4L_4
\ge4,
$$
because $t-2a+3c\ge0$.

Let $L_3$ be the number of radius-$3$
steps. After these steps, $S=a$ and
$$
W=4(c+L_4)+3(a-c+L_3).
$$
Thus, when we begin using radius $2$,
the upper bound on $s$ in
\eqref{eq:highorder-step} is
$$
2(t+1)-2r+1+W-S
=c+4L_4+3L_3+3
\ge3.
$$

For a fixed radius $\ell$, each step
increases $W-S$ by
$\ell(s+1)-s=(\ell-1)s+\ell>0$.
Hence, the upper bound on $s$ remains
positive during the steps with that
radius. The rule
\eqref{eq:phase-rule} therefore reaches
total nullity $t$.

Let $L_2$ be the number of radius-$2$
steps. The steps of radii $4$, $3$,
and $2$ increase the nullity by
$c$, $a-c$, and $t-a$, respectively.
Their total support bound, in the
notation of
Theorem~\ref{thm:successive}, is
$$
W=4(c+L_4)+3(a-c+L_3)+2(t-a+L_2)
=2t+a+c+4L_4+3L_3+2L_2.
$$
Appendix~\ref{app:iteration-counts}
proves that the steps chosen by
\eqref{eq:phase-rule} satisfy
\begin{equation}\label{eq:three-phase-cost}
W=2t+a+c+4L_4+3L_3+2L_2
\le2t+2a+2=B.
\end{equation}

We still have to verify that
Lemma~\ref{lem:ball-exponent} applies
to every step. Its requirements are
$\ell t\ge r$ and $12\ell\le r$.
The first holds because $\ell\ge2$
and $r\le2t$.

For the second, note that
$r=t+a\ge42>36=12\cdot3$, which covers the radii
$3$ and $2$. If radius $4$ is used, then $c>0$,
so $a>t/2$ and
$r=t+a>3t/2\ge48=12\cdot4$.
Thus, $12\ell\le r$ for every
radius used.

The choice of $s$ in
\eqref{eq:phase-rule} ensures
\eqref{eq:highorder-step} at every
step. Lemma~\ref{lem:ball-exponent}
therefore gives the uniform volume
conditions in
Theorem~\ref{thm:successive}.
The steps reach total nullity $t$,
and \eqref{eq:three-phase-cost}
gives $W\le B$. Hence,
$d_t(\C)\le B$, contrary to
our assumption.
\end{proof}

We next consider $2t\le r\le3t$. In this
range we use two puncturing steps: the
first produces nullity $s$, and the second
produces the remaining nullity $t-s$.

\begin{proposition}\label{prop:middle-ratio}
Let $\C$ be an $[n,k]_q$ linear code with
redundancy $\rho$, let
$32\le t\le\min\{k,\rho\}$, and set
$r:=R_t(\C)$. If $2t\le r\le3t$, then
$d_t(\C)\le2r+2$.
\end{proposition}

\begin{proof}
Suppose, for a contradiction, that
$d_t(\C)\ge2r+3$. We choose the Hamming
radii of the two steps according to
$r/t$:
$$
\begin{array}{c|ccc}
\text{Range of }r/t
 &[2,9/4]&(9/4,8/3)&[8/3,3]\\ \hline
\text{Radii }(\ell_1,\ell_2)
 &(5,3)&(6,3)&(7,4).
\end{array}
$$
For the chosen pair $(\ell_1,\ell_2)$, set
$$
z:=\left\lceil
\frac{2r-(\ell_2-1)t-\ell_2-1}{\ell_1}
\right\rceil,
\qquad s:=z-1.
$$
The intervals in the table give
$2\le z\le t$. Indeed, in the three rows the numerator
$2r-(\ell_2-1)t-\ell_2-1$ is at least $2t-4$, $5t/2-4$ and
$7t/3-5$, respectively, which exceeds $\ell_1$ for $t\ge32$;
and it is at most $5t/2$, $10t/3$ and $3t$, respectively,
which is at most $\ell_1t$. Thus, $1\le s<t$,
and we may apply
Theorem~\ref{thm:successive} with
$(\ell_1,s)$ followed by
$(\ell_2,t-s)$. These steps have
support bounds $\ell_1z$ and
$\ell_2(t-s+1)$, respectively.

We first check the volume conditions.
For the first step, the weaker estimate
in Lemma~\ref{lem:ball-exponent} gives
the required inequality provided that
\begin{equation}\label{eq:middle-first}
s-1\le\ell_1(t+1)-2r.
\end{equation}
Before the second step, the accumulated
nullity is $s$ and the accumulated
support bound is $\ell_1z$. The
corresponding requirement is
$$
t-\ell_1z-1\le\ell_2(t+1)-2r.
$$
This is equivalent to
$$
\ell_1z\ge
2r-(\ell_2-1)t-\ell_2-1,
$$
which holds by the definition of $z$.
The choice of $z$ therefore ensures
the volume condition for the second
step. We still have to check
\eqref{eq:middle-first}.

Because the numerator defining $z$
is an integer,
$$
z\le
\frac{2r-(\ell_2-1)t-\ell_2
+\ell_1-2}{\ell_1}.
$$
Using this bound and $s=z-1$,
we obtain \eqref{eq:middle-first}
whenever
\begin{equation}\label{eq:middle-check-first}
2(\ell_1+1)r
\le(\ell_1^2+\ell_2-1)t
+\ell_1^2+\ell_1+\ell_2+2.
\end{equation}
Since $\ell_1>\ell_2$ in every row of the table, the quantity
$\ell_1z+\ell_2(t-s+1)=(\ell_1-\ell_2)z+\ell_2(t+2)$ is increasing
in $z$. The same bound on $z$ therefore gives
$\ell_1z+\ell_2(t-s+1)\le2r+2$
whenever
\begin{equation}\label{eq:middle-check-support}
2\ell_2r
\ge(\ell_1+\ell_2^2-\ell_2)t
+\ell_1^2+\ell_2^2-4\ell_1+2\ell_2.
\end{equation}

For each pair of radii in the table, both inequalities are affine in $r$.
Since the left-hand side of \eqref{eq:middle-check-first} increases
with $r$, this inequality is strongest at the upper endpoint of the
corresponding interval, while \eqref{eq:middle-check-support} is
strongest at the lower endpoint. We call the \emph{margin} of an
inequality the amount by which its larger side exceeds its smaller side.
At the three upper endpoints, the margins in
\eqref{eq:middle-check-first} are
$$
35,\qquad \frac{2t}{3}+47,\qquad
4t+62.
$$
At the three lower endpoints, the margins in
\eqref{eq:middle-check-support} are
$$
t-20,\qquad \frac{3t}{2}-27,\qquad
\frac{7t}{3}-45.
$$
They are all nonnegative for $t\ge32$. Hence, the first volume
condition holds and the two support bounds sum to at most $2r+2$.

Finally, we check the hypotheses
of Lemma~\ref{lem:ball-exponent} for
both radii. Since $r\ge2t$, we have
$a=r-t\ge t\ge32$, so the lemma applies. The intervals give
$\ell_1t\ge r$ and $\ell_2t\ge r$.
They also give
$12\max\{\ell_1,\ell_2\}\le r$:
at the lower endpoints, this amounts
to $60\le2t$, $72\le9t/4$,
and $84\le8t/3$, respectively.
Thus, the lemma applies at both steps.
Theorem~\ref{thm:successive} gives
$d_t(\C)\le2r+2$, contrary to our
assumption.
\end{proof}

We finally consider $r\ge3t$. We again
use two puncturing steps, dividing the
required nullity as nearly in half as
possible.

\begin{proposition}\label{prop:large-ratio}
Let $\C$ be an $[n,k]_q$ linear code with
redundancy $\rho$, let
$32\le t\le\min\{k,\rho\}$, and set
$r:=R_t(\C)$. If $r\ge3t$, then
$d_t(\C)\le2r+2$.
\end{proposition}

\begin{proof}
Suppose, for a contradiction, that
$d_t(\C)\ge2r+3$. Set
$s:=\lfloor t/2\rfloor$ and $u:=t-s$.
We apply Theorem~\ref{thm:successive}
with $(\ell_1,s)$ followed by
$(\ell_2,u)$.

Lemma~\ref{lem:ball-exponent}
gives the volume condition for the
first step if
$s-1\le\ell_1(t+1)-2r$.
We therefore choose
$$
\ell_1:=
\left\lceil\frac{2r+s-1}{t+1}\right\rceil
$$
and set $w:=\ell_1(s+1)$, the
support bound for that step. Before
the second step, the accumulated
nullity and support bound are $s$
and $w$. Its volume condition
follows from the weaker estimate
of Lemma~\ref{lem:ball-exponent}
provided that
$$
u-1-w+s=t-w-1
\le\ell_2(t+1)-2r.
$$
In addition, the lemma requires
$\ell_2t\ge r$. We satisfy both
inequalities by taking
$$
\ell_2:=\max\left\{
\left\lceil\frac rt\right\rceil,
\left\lceil\frac{2r-w+t-1}{t+1}\right\rceil
\right\}.
$$

We check the other hypotheses of
Lemma~\ref{lem:ball-exponent}, which can be used because $a=r-t\ge2t\ge64$.
First, $\ell_1t\ge r$. Indeed,
$\ell_1(t+1)\ge2r+s-1$, and
$$
t(2r+s-1)-r(t+1)
=r(t-1)+t(s-1)\ge0.
$$
We also have $\ell_2\le\ell_1$.
The first term defining $\ell_2$
is at most $\ell_1$ because
$\ell_1t\ge r$. For the second
term, $\ell_1(t+1)\ge2r+s-1$
and $w=\ell_1(s+1)\ge u$ imply
$$
2r-w+t-1\le\ell_1(t+1).
$$
Hence,
$\lceil(2r-w+t-1)/(t+1)\rceil\le\ell_1$.

Since $2r+s-1$ is an integer,
$$
\ell_1
\le\frac{2r+s+t-1}{t+1}
\le\frac{2r}{t}+2
\le\frac r{12}.
$$
The final inequality holds for
$t\ge32$ and $r\ge3t$.
Consequently $12\ell_1\le r$
and $12\ell_2\le r$. All the
hypotheses of
Lemma~\ref{lem:ball-exponent}
have now been checked for both
steps.

It remains to show that their
support bounds sum to at most
$2r+2$, that is,
\begin{equation}\label{eq:large-ratio-support}
w+\ell_2(u+1)\le2r+2.
\end{equation}
The bound on $\ell_1$ above gives
$$
w\le
\frac{(s+1)(2r+s+t-1)}{t+1}.
$$
We consider the two terms in the definition of $\ell_2$ separately,
removing the ceiling in each case by a simple upper bound.

If $\ell_2=\lceil r/t\rceil$, then
$\lceil r/t\rceil\le(r+t-1)/t$,
and hence
$w+\ell_2(u+1)\le U_1(2r)$, where
$$
U_1(y):=
\frac{(s+1)(y+s+t-1)}{t+1}
+\frac{(u+1)(y/2+t-1)}{t}.
$$
If instead
$$
\ell_2=
\left\lceil\frac{2r-w+t-1}{t+1}\right\rceil,
$$
then
$w+\ell_2(u+1)\le U_2(2r)$, where
$$
U_2(y):=
\frac{s(s+1)(y+s+t-1)}{(t+1)^2}
+\frac{(u+1)(y+2t-1)}{t+1}.
$$
For this second estimate, the
coefficient of $w$ is
$1-(u+1)/(t+1)=s/(t+1)>0$,
so we may substitute the upper
bound on $w$.

Both $U_1(y)$ and $U_2(y)$ are
affine functions of $y$ whose
coefficients of $y$ are smaller
than $1$. As $2r\ge6t$, it is
enough to check
$U_1(6t)\le6t+2$ and
$U_2(6t)\le6t+2$. Write
$t=2h$ or $t=2h+1$, so that
$h\ge16$. After clearing
positive denominators, these
four inequalities reduce to the
positivity of
$$
\begin{array}{c|cc}
& U_1 & U_2\\ \hline
t=2h
&2h^3-18h^2+h+1
&h^3-4h^2+8h+3\\
t=2h+1
&2h^2-9h-2
&h^3-3h^2+4.
\end{array}
$$
All four expressions are positive
for $h\ge16$. This proves
\eqref{eq:large-ratio-support}.

The two steps therefore satisfy
the volume conditions of
Theorem~\ref{thm:successive},
have total nullity $s+u=t$,
and have total support bound at
most $2r+2$. The theorem gives
$d_t(\C)\le2r+2$, contrary to
our assumption.
\end{proof}

Propositions~\ref{prop:large-diagonal},
\ref{prop:middle-ratio}, and
\ref{prop:large-ratio} cover the range
$r-t\ge6$. Together with
Theorem~\ref{thm:fivegaps}, they give
the result for every $t\ge32$.

\begin{theorem}\label{thm:highorder}
Let $\C$ be an $[n,k]_q$ linear
code with redundancy $\rho$, and let
$32\le t\le\min\{k,\rho\}$. Then
$$
d_t(\C)\le2R_t(\C)+2.
$$
\end{theorem}

\begin{proof}
Set $r:=R_t(\C)$. If $r-t\le5$,
the result follows from
Theorem~\ref{thm:fivegaps}.
Otherwise $r-t\ge6$, and we apply
Proposition~\ref{prop:large-diagonal}
when $r\le2t$,
Proposition~\ref{prop:middle-ratio}
when $2t\le r\le3t$, or
Proposition~\ref{prop:large-ratio}
when $r\ge3t$.
\end{proof}

\subsection{The case \texorpdfstring{$3\le t\le31$}{3 <= t <= 31}}\label{sec:low}

It remains to consider the finitely many orders $3\le t\le31$. Our goal in
this subsection is to obtain, for each such $t$, an explicit upper bound on
the radius $r$ of a possible counterexample. The case $t=3$ requires some
separate estimates, which we treat first. We then show that whenever
$r>2t$, a counterexample must satisfy $\rho>2r$, and use this extra
redundancy to derive sharper Hamming-ball estimates. These estimates yield
uniform large-radius bounds for $4\le t\le31$. We finally return to $t=3$
and treat its remaining large-radius ranges.

\subsubsection{A preliminary estimate for \texorpdfstring{$t=3$}{t=3}}

We begin with an estimate specific to $t=3$. Besides reducing the
remaining range for $t=3$, it will also be used in the general redundancy
bound below.

For the proof we use the following exact recurrence for Hamming-ball
volumes:
\begin{equation}\label{eq:recurrence}
\vol_Q(n+2,r+1)-Q^2\vol_Q(n,r)
=
(Q-1)^{r+1}
\left(
\binom n{r+1}-(Q-1)\binom nr
\right).
\end{equation}
Indeed, separating the last two coordinates gives
$$
\vol_Q(n+2,r+1)
=\vol_Q(n,r+1)+2(Q-1)\vol_Q(n,r)
 +(Q-1)^2\vol_Q(n,r-1).
$$
Subtracting $Q^2\vol_Q(n,r)$ and using
$\vol_Q(n,j)-\vol_Q(n,j-1)=\binom nj(Q-1)^j$
gives \eqref{eq:recurrence}.

\begin{theorem}\label{thm:order3}
Let $\C$ be an $[n,k]_q$ linear code with redundancy $\rho$,
let $\min\{k,\rho\}\ge 3$, and set $r:=R_3(\C)$.
\begin{enumerate}[label=(\alph*)]
    \item Over every finite field, if $\rho\le2r+2$, then
    $d_3(\C)\le2r+2$.
    \item Over the binary field, the same conclusion holds if
    $\rho\le2r+3$.
\end{enumerate}
\end{theorem}

\begin{proof}
We prove (b) first. Suppose that $q=2$ and that $\C$ is a counterexample. By \eqref{eq:failure-range}, $\rho\ge2r$.
Thus, $\varepsilon:=\rho-2r$ belongs to $\{0,1,2,3\}$. Set
$b_\varepsilon:=2^{\varepsilon+3}-3$. By Eq.~\eqref{eq:binary-length},
$n\le2r+b_\varepsilon$.

For fixed $\varepsilon$, put
$F_\varepsilon(r):=\vol_8(2r+b_\varepsilon,r)/8^{2r+\varepsilon}$.
By Eq.~\eqref{eq:recurrence}, the sign of
$F_\varepsilon(r+1)-F_\varepsilon(r)$ is the sign of
$b_\varepsilon-7-6r$. Since $r\ge3$, the largest value of
$F_\varepsilon(r)$ occurs at $r=3,3,4,9$ for
$\varepsilon=0,1,2,3$, respectively. When $\varepsilon=3$,
the same value is also attained at $r=10$. We compare these
values below.
\begin{center}
\small
\begin{tabular}{rrrrr}
\hline
$\varepsilon$ & $r$ & $2r+b_\varepsilon$
& $\vol_8(2r+b_\varepsilon,r)$ & $8^{2r+\varepsilon}$\\
\hline
0&3&11&59\,368&262\,144\\
1&3&19&340\,880&2\,097\,152\\
2&4&37&161\,272\,049&1\,073\,741\,824\\
3&9&79&8\,458\,053\,920\,168\,206\,726
&9\,223\,372\,036\,854\,775\,808\\
\hline
\end{tabular}
\end{center}
In every row the fourth entry is smaller than the fifth. Hence,
$F_\varepsilon(r)<1$ for every admissible $r$. Since
$n\le2r+b_\varepsilon$ and the volume is increasing in the length,
$\vol_8(n,r)\le\vol_8(2r+b_\varepsilon,r)<8^{\rho}$, contradicting
\eqref{eq:cover}.

\medskip

We now prove (a). The binary case follows from (b), so assume
$q\ge3$. If $q\ge r$, Proposition~\ref{prop:inputs} gives
$d_3(\C)\le2r+1$. We may therefore assume $r\ge q+1$.
Suppose that $\C$ is a counterexample with $\rho\le2r+2$,
and set $\varepsilon:=\rho-2r$ and $Q:=q^3$. By \eqref{eq:failure-range},
$\rho\ge2r$, so $\varepsilon\in\{0,1,2\}$.
By Proposition~\ref{prop:weighted-length},
$n\le2r+b$, where $b:=U_q(\varepsilon+3,4)-3$.
Eq.~\eqref{eq:U} gives $U_3(h,4)=P_h$ and
$U_q(h,4)=1+3P_{h-1}$ for $q\ge4$.

Set $F(r):=\vol_Q(2r+b,r)/Q^{2r+\varepsilon}$.
By Eq.~\eqref{eq:recurrence}, the sign of
$F(r+1)-F(r)$ is the sign of
$b-(Q-1)-(Q-2)r$. This is negative for $r\ge q+1$.
Indeed, for $q=3$ we have $b\le118<26+25\cdot4$.
For $q\ge4$, we have $b\le3q^3+3q^2+3q+1$,
and the claim follows from
$(q-1)q^3-3q^2-5q-4>0$.

It is therefore enough to prove $F(q+1)<1$. Put
$r_0:=q+1$ and $n_0:=2r_0+b$. For $q\ge5$,
the values of $n_0$ satisfy
\begin{center}
\begin{tabular}{rcc}
\hline
$\varepsilon$ & $n_0$ & Upper bound\\
\hline
0 & $5q+3$ & $Q$\\
1 & $3q^2+5q+3$ & $Q$\\
2 & $3q^3+3q^2+5q+3$ & $4Q$\\
\hline
\end{tabular}
\end{center}
We use
$\vol_Q(n_0,r_0)\le\binom{n_0}{r_0}Q^{r_0}$.
For $\varepsilon=0,1$, the table gives $n_0<Q$, and hence
$\binom{n_0}{r_0}\le n_0^{r_0}/r_0!<Q^{r_0}$.
This implies $F(r_0)<1$. For $\varepsilon=2$ and
$q\ge11$, we have $n_0<4Q$ and $r_0\ge12$, so
$\binom{n_0}{r_0}<
(\mathrm e n_0/r_0)^{r_0}<Q^{r_0}$,
since $\mathrm e<3$. Thus, $F(r_0)<1$ in this case as well.

The cases not covered by these estimates are listed below. In
each case, exact integer evaluation gives
$\vol_{q^3}(n_0,r_0)<q^{3(2r_0+\varepsilon)}$.
\begin{center}
\small
\begin{tabular}{rrrr@{\qquad}|rrrr}
\hline
$q$&$\varepsilon$&$r_0$&$n_0$
&$q$&$\varepsilon$&$r_0$&$n_0$\\
\hline
3&0&4&18&3&2&4&126\\
3&1&4&45&4&2&5&263\\
4&0&5&23&5&2&6&478\\
4&1&5&71&7&2&8&1214\\
&&&&8&2&9&1771\\
&&&&9&2&10&2478\\
\hline
\end{tabular}
\end{center}
These comparisons are also checked by the programs described in
Appendix~\ref{app:computations}. We have proved $F(q+1)<1$,
and the monotonicity of $F$ gives $F(r)<1$ for every
$r\ge q+1$. Since $n\le2r+b$ and the volume is increasing in the
length, $\vol_Q(n,r)\le\vol_Q(2r+b,r)<Q^{\rho}$, contradicting
\eqref{eq:cover}.
\end{proof}

\subsubsection{Redundancy of a putative counterexample}

We next show that, if the generalized covering radius $r=R_t(\C)$
is larger than $2t$, the redundancy of a counterexample must satisfy
$\rho>2r$. We will use this inequality to bound the size of Hamming
balls after puncturing $\C$, as required by Theorem~\ref{thm:successive}.

\begin{lemma}\label{lem:low-ratio}
Let $\C$ be an $[n,k]_q$ linear code with redundancy $\rho$, let
$3\le t\le\min\{k,\rho\}$, and set $r:=R_t(\C)$. If $\C$ is a
counterexample for order $t$ and $r>2t$, then $\rho>2r$.
\end{lemma}

\begin{proof}
For $t=3$, the result follows from Theorem~\ref{thm:order3}. Suppose therefore
that $t\ge4$ and, by contradiction, that
$$
\lambda:=\frac{\rho}{r}\le2.
$$
Set $a:=r-t$ and define
$$
m:=2r-t+2
\qquad \textnormal{and} \qquad
\gamma:=\rho-m\ge1,
$$
as in \eqref{eq:gamma}, and let $b$ be the least positive integer such that $q^b\ge t$. Then
$b\le t-1$.

Applying Lemma~\ref{lem:incidence} with the parameter $m-b$ gives
\begin{equation}\label{eq:capratio}
n\le m+(t+b+1)q^\gamma.
\end{equation}
Indeed,
$$
\frac{t+b-1}{1-q^{-b}}
\le
t+\frac{bt}{t-1}=t+b+\frac{b}{t-1}
\le t+b+1,
$$
where the first inequality follows from $q^b\ge t$, and the last from
$b\le t-1$.

On the other hand, Eq.~\eqref{eq:coverrough} gives
\begin{equation}\label{eq:coverratio}
nq^{-\gamma}>
\frac r{\mathrm e}q^{2+a(2-\lambda)}.
\end{equation}
Since
$$
\gamma=\rho-(2r-t+2)=r(\lambda-2)+t-2,
$$
we also have
$$
a(2-\lambda)=\frac{a(t-2-\gamma)}{r}.
$$

Suppose first that $\gamma\ge3$. Since $q^\gamma\ge8$, $m<2r$, and
$t+b+1\le2t<r$, Eq.~\eqref{eq:capratio} gives
$$
nq^{-\gamma}
<
\frac r4+r
=
\frac{5r}{4}.
$$
In contrast, Eq.~\eqref{eq:coverratio} gives
$$
nq^{-\gamma}>
\frac{4r}{\mathrm e}>
\frac{16r}{11},
$$
a contradiction.

Suppose next that $\gamma=1$. Since $r>2t$, we have $a/r>1/2$, while
$t-3\ge1$. Hence,
$$
2+a(2-\lambda)
=
2+\frac{a(t-3)}r
>
\frac52.
$$
Therefore, Eq.~\eqref{eq:coverratio} gives
$$
nq^{-1}>
\frac{2^{5/2}r}{\mathrm e}>2r.
$$
On the other hand, since $m=2r-t+2$, $t+b+1\le2t$ and $r\ge2t+1$,
Eq.~\eqref{eq:capratio} gives
$$
nq^{-1}\le\frac m2+t+b+1\le r+\frac{3t}{2}+1<2r,
$$
again a contradiction.

Now, let $\gamma=2$ and $t\ge5$. As above,
$$
2+a(2-\lambda)
=
2+\frac{a(t-4)}r
>
\frac52,
$$
so Eq.~\eqref{eq:coverratio} gives $nq^{-2}>2r$. But
Eq.~\eqref{eq:capratio} gives
$$
nq^{-2}
<
\frac r2+r
=
\frac{3r}{2},
$$
which is impossible.

It remains to treat the case $t=4$ and $\gamma=2$. Here $b\le2$ and, since
$r>8$, we have $r\ge9$. Thus, Eq.~\eqref{eq:capratio} gives
$$
nq^{-2}
\le
\frac r2+7
<
\frac{16r}{11},
$$
whereas Eq.~\eqref{eq:coverratio} gives
$$
nq^{-2}>
\frac{4r}{\mathrm e}>
\frac{16r}{11}.
$$
This is again a contradiction.

All possible values of $\gamma$ have been excluded. Hence, $\rho>2r$.
\end{proof}

\subsubsection{Hamming balls after puncturing}\label{sec:ball}

For $t\le31$, the uniform estimate used in the large-order case is no longer
strong enough in all the ranges that remain. We therefore derive a sharper
lower bound for ordinary Hamming balls after puncturing. The argument uses
the covering lower bound on $n$, together with an elementary estimate for
binomial coefficients.

We begin with the following elementary estimate.

\begin{lemma}\label{lem:factorial}
For every integer $j\ge1$,
$$
j!\le \mathrm e\,j\left(\frac{j}{\mathrm e}\right)^j
<3j\left(\frac{j}{\mathrm e}\right)^j.
$$
Moreover, for every $N\ge j$,
$$
\binom Nj\ge\frac{(N-j+1)^j}{j!}.
$$
\end{lemma}

\begin{proof}
The lower bound for the binomial coefficient follows immediately from the product formula
$$
\binom Nj
=
\prod_{i=0}^{j-1}\frac{N-i}{j-i}
\ge
\frac{(N-j+1)^j}{j!}.
$$

For the first inequality, the monotonicity of $\log x$ gives
$$
\log(j!)
\le
\int_1^j\log x\,dx+\log j
=
(j+1)\log j-j+1.
$$
Exponentiating, we obtain
$$
j!
\le
\mathrm e\,j\left(\frac{j}{\mathrm e}\right)^j.
$$
The strict inequality follows from $\mathrm e<3$.
\end{proof}

We now turn the lower bound on $n$ from Eq.~\eqref{eq:coverrough} into a
lower bound for ordinary Hamming balls after $W$ coordinates have been
punctured. The quantity $\beta_q(W,\ell)$ introduced below measures the loss
caused by puncturing and by the factorial estimate. When it is sufficiently
large, the resulting bound takes the simple exponential form needed in the
successive puncturing argument.

\begin{lemma}\label{lem:sharp}
Let $\C$ be an $[n,k]_q$ linear code with redundancy $\rho$, let
$t\in[\min\{k,\rho\}]$, and set $r:=R_t(\C)$. Suppose that
$$\lambda:=\frac{\rho}{r}>2.
$$
For integers $W\ge0$ and $\ell\ge1$, define
\begin{equation}\label{eq:beta}
\beta_q(W,\ell):=
\left(1-\frac1q\right)
\left(
\frac r\ell-\frac{11(W+\ell-1)}{4\ell q^t}
\right).
\end{equation}
If $\beta_q(W,\ell)>0$, then $n-W\ge \ell$ and
\begin{equation}\label{eq:sharpball}
\vol_q(n-W,\ell)>
\frac{\beta_q(W,\ell)^\ell}{3\ell}\,
q^{\ell t(\lambda-1)+\ell}.
\end{equation}

Suppose in addition that $\ell t\ge r$ and that
$\beta_q(W,\ell)\ge c$ for a real number $c$ for which there is an
integer $j\ge1$ with
\begin{equation}\label{eq:cj}
\ell\ge j,\qquad
c^j\ge3j,\qquad
c\ge\frac{j+1}{j}.
\end{equation}
Then
\begin{equation}\label{eq:simpleexponent}
\vol_q(n-W,\ell)>
q^{\rho+\ell(t+1)-2r}.
\end{equation}

For fixed $t,r,W,\ell$, the same conclusions hold for every
$q\ge q_0$ if the corresponding conditions on $\beta$ are checked
with $\beta_{q_0}(W,\ell)$ in place of $\beta_q(W,\ell)$.
\end{lemma}

\begin{proof}
By Eq.~\eqref{eq:coverrough},
$$
n>\frac r{\mathrm e}q^{t(\lambda-1)}.
$$
Since $\lambda>2$, we also have
$q^{t(\lambda-1)}>q^t$. Using $\mathrm e<11/4$, it follows that
$$
n-W-\ell+1>
\frac r{\mathrm e}q^{t(\lambda-1)}
\left(
1-\frac{11(W+\ell-1)}{4rq^t}
\right).
$$
If $\beta_q(W,\ell)>0$, the expression in parentheses is positive, and hence
$n-W-\ell+1>0$. Since this is an integer, $n-W\ge\ell$.

By Lemma~\ref{lem:factorial},
$$
\binom{n-W}{\ell}
\ge
\frac{(n-W-\ell+1)^\ell}{\ell!}
>
\frac1{3\ell}
\left(
\frac{\mathrm e(n-W-\ell+1)}{\ell}
\right)^\ell.
$$
Therefore,
\begin{align*}
\vol_q(n-W,\ell)
&\ge
\binom{n-W}{\ell}(q-1)^\ell\\
&>
\frac1{3\ell}
\left(
\frac{\mathrm e(n-W-\ell+1)(q-1)}{\ell}
\right)^\ell\\
&>
\frac{\beta_q(W,\ell)^\ell}{3\ell}
q^{\ell t(\lambda-1)+\ell},
\end{align*}
which proves Eq.~\eqref{eq:sharpball}.

Suppose now that $\ell t\ge r$ and $\beta_q(W,\ell)\ge c$, with
$$
\ell\ge j,\qquad
c^j\ge3j,\qquad
c\ge\frac{j+1}{j}.
$$
These conditions imply
$$
c^\ell\ge3\ell
$$
for every $\ell\ge j$. Indeed, this holds for $\ell=j$, and if
$c^\ell\ge3\ell$, then
$$
c^{\ell+1}\ge3\ell c\ge3(\ell+1),
$$
since
$$
c\ge\frac{j+1}{j}\ge\frac{\ell+1}{\ell}.
$$
Hence,
$$
\frac{\beta_q(W,\ell)^\ell}{3\ell}\ge1.
$$

Moreover,
\begin{align*}
\ell t(\lambda-1)+\ell-\rho
&=(\ell t-r)\lambda-\ell t+\ell\\
&\ge2(\ell t-r)-\ell t+\ell\\
&=\ell(t+1)-2r,
\end{align*}
where we used $\lambda>2$ and $\ell t\ge r$. Together with the strict
inequality in Eq.~\eqref{eq:sharpball}, this gives
$$
\vol_q(n-W,\ell)>
q^{\rho+\ell(t+1)-2r},
$$
which proves Eq.~\eqref{eq:simpleexponent}.

Finally, fix $t,r,W,\ell$ and suppose that
$\beta_{q_0}(W,\ell)>0$. Both factors in the definition of
$\beta_q(W,\ell)$ are then positive and nondecreasing for
$q\ge q_0$. Hence,
$\beta_q(W,\ell)\ge\beta_{q_0}(W,\ell)$ for every such $q$,
which proves the last assertion.
\end{proof}

For the finite parameter range considered later, we will use an
exact rational version of the preceding estimate. It avoids the
factorial approximation in Lemma~\ref{lem:factorial} and allows the
required inequalities to be checked using integer arithmetic.

\begin{lemma}\label{lem:K}
Let $\C$ be an $[n,k]_q$ linear code with redundancy $\rho$, let
$t\in[\min\{k,\rho\}]$, set $r:=R_t(\C)$, and suppose that
$\lambda:=\rho/r>2$. Let $2\le q_0\le q$, and let $W\ge0$ and $\ell\ge1$
be integers with $\ell t\ge r$. Define
\begin{equation}\label{eq:Ksymbolic}
K_{q_0}(W,\ell):=
\left(1-\frac1{q_0}\right)
\left(
\frac{4r}{11}-\frac{W+\ell-1}{q_0^t}
\right).
\end{equation}
If
$$
K_{q_0}(W,\ell)>0
\qquad\text{and}\qquad
K_{q_0}(W,\ell)^\ell\ge\ell!,
$$
then $n-W\ge\ell$ and
$$
\vol_q(n-W,\ell)>
q^{\rho+\ell(t+1)-2r}.
$$
\end{lemma}

\begin{proof}
By \eqref{eq:coverrough} and $\mathrm e<11/4$, we have
$n>\frac{4r}{11}q^{t(\lambda-1)}$. Since $\lambda>2$, also
$q^{t(\lambda-1)}>q^t$, and hence
$W+\ell-1\le\frac{W+\ell-1}{q^t}\,q^{t(\lambda-1)}$. Therefore,
$$
n-W-\ell+1>
\left(
\frac{4r}{11}-\frac{W+\ell-1}{q^t}
\right)
q^{t(\lambda-1)}.
$$
Both $1-1/q$ and $\frac{4r}{11}-\frac{W+\ell-1}{q^t}$ are nondecreasing
in $q$, and both are positive at $q=q_0$ because $K_{q_0}(W,\ell)>0$.
Multiplying by $q-1=q(1-1/q)$ therefore gives
$$
(n-W-\ell+1)(q-1)>K_{q_0}(W,\ell)\,q^{t(\lambda-1)+1}>0.
$$
In particular $n-W-\ell+1>0$, so $n-W\ge\ell$. The product formula of
Lemma~\ref{lem:factorial} now gives
$$
\vol_q(n-W,\ell)
\ge\binom{n-W}{\ell}(q-1)^\ell
\ge\frac{\bigl((n-W-\ell+1)(q-1)\bigr)^\ell}{\ell!}
>\frac{K_{q_0}(W,\ell)^\ell}{\ell!}\,q^{\ell(t(\lambda-1)+1)}
\ge q^{\ell(t(\lambda-1)+1)}.
$$
Finally, as in the proof of Lemma~\ref{lem:sharp}, $\lambda>2$ and
$\ell t\ge r$ give
$\ell(t(\lambda-1)+1)-\rho=(\ell t-r)\lambda-\ell t+\ell\ge\ell(t+1)-2r$.
\end{proof}

After clearing denominators, the hypotheses of Lemma~\ref{lem:K} can be
checked using integer arithmetic.

We now use these estimates to obtain an explicit large-radius threshold for each order $4\le t\le31$.

\subsubsection{The case of large \texorpdfstring{$r$}{r} for \texorpdfstring{$4\le t\le31$}{4 <= t <= 31}}

We now apply Lemma~\ref{lem:sharp} to obtain, for each fixed order
$4\le t\le31$, an explicit threshold above which the conjecture holds.
Since all the thresholds considered below satisfy $r>2t$, Lemma~\ref{lem:low-ratio}
ensures that any counterexample in this range satisfies $\rho>2r$.

As in the proof of Proposition~\ref{prop:large-ratio}, we use two
successive puncturing steps. We choose their parameters so that the
inequalities required by Theorem~\ref{thm:successive}, together with
the bound on the total support, follow from affine inequalities in
$r$. For each $t$, we check these inequalities at $r=r_0(t)$
and show that their slopes are nonnegative. They then hold for every
$r\ge r_0(t)$.

Fix $4\le t\le31$, choose an integer $s\in[t]$, and set
$$
u:=t-s.
$$
For an integer $r$, define
\begin{equation}\label{eq:tailchain}
\ell:=\left\lceil\frac{2r+s-1}{t+1}\right\rceil,
\qquad
w:=(s+1)\ell,
\qquad
f:=\max\left\{
\left\lceil\frac rt\right\rceil,
\left\lceil\frac{2r-w+t-1}{t+1}\right\rceil
\right\}.
\end{equation}

We apply Theorem~\ref{thm:successive} in two steps, with
$(\ell_1,s_1)=(\ell,s)$ and $(\ell_2,s_2)=(f,u)$. The first step
punctures $w=\ell(s+1)$ coordinates. By the definitions of
$\ell$ and $f$, we have $ft\ge r$ and $\ell t\ge r$ (the latter because
$t(2r+s-1)-r(t+1)=r(t-1)+t(s-1)\ge0$, as in the proof of
Proposition~\ref{prop:large-ratio}), and
$$
\ell(t+1)-2r\ge s-1,
\qquad
f(t+1)-2r\ge t-w-1.
$$
Thus, whenever Lemma~\ref{lem:sharp} gives
\eqref{eq:simpleexponent} for both steps, the resulting volume
bounds imply the inequalities required by
Theorem~\ref{thm:successive}.

The ceilings in \eqref{eq:tailchain} make it inconvenient to bound the
two steps directly as functions of $r$. We use the following affine
lower and upper bounds for $\ell$ and $w$, and upper bounds for the
two terms defining $f$:
\begin{align*}
\ell_-&:=\frac{2r+s-1}{t+1},
&
\ell_+&:=\frac{2r+s+t-1}{t+1},
\\
w_-&:=(s+1)\ell_-,
&
w_+&:=(s+1)\ell_+,
\\
f_A&:=\frac{r+t-1}{t},
&
f_B&:=\frac{2r-w_-+2t-1}{t+1}.
\end{align*}
Indeed, $\ell_-\le\ell\le\ell_+$ and
$w_-\le w\le w_+$, while
$f\le\max\{f_A,f_B\}$.

The two steps puncture $w+(u+1)f$ coordinates. If the term
defining $f_A$ is larger, this number is at most
$U_A:=w_++(u+1)f_A$. For the other term, we first write
$$
w+\frac{u+1}{t+1}(2r-w+2t-1)
=
\frac{s}{t+1}w+\frac{(u+1)(2r+2t-1)}{t+1}.
$$
The coefficient $s/t+1$ is positive, so this expression is at most
$$
U_B:=\frac{s}{t+1}w_+
+\frac{(u+1)(2r+2t-1)}{t+1}.
$$
Thus, the two steps puncture at most $\max\{U_A,U_B\}$
coordinates.

To apply Lemma~\ref{lem:sharp} over every finite field, we use
$q_0=2$. For $c>0$, define
$$
\Phi_c(W,E):=
\frac12\left(
r-\frac{11(W+E-1)}{4\cdot2^t}
\right)-cE.
$$
By Eq.~\eqref{eq:beta}, we have
$\Phi_c(W,E)=E(\beta_2(W,E)-c)$. Thus,
$\Phi_c(W,E)\ge0$ implies $\beta_q(W,E)\ge c$
for every $q\ge2$. The function $\Phi_c(W,E)$ decreases
as either $W$ or $E$ increases, so we can use the upper
bounds on $\ell$, $w$, and $f$ obtained above.

Choose $c_1,j_1,c_2,j_2$ such that
$c_i^{j_i}\ge3j_i$ and
$c_i\ge(j_i+1)/j_i$ for $i=1,2$.
It is then sufficient to check
\begin{equation}\label{eq:fivemargins}
\begin{aligned}
2r+2-U_A&\ge0,\\
2r+2-U_B&\ge0,\\
\Phi_{c_1}(0,\ell_+)&\ge0,\\
\Phi_{c_2}(w_+,f_A)&\ge0,\\
\Phi_{c_2}(w_+,f_B)&\ge0,
\end{aligned}
\end{equation}
together with $\ell_-\ge j_1$ and $r/t\ge j_2$.

The first two inequalities bound the number of coordinates
punctured in the two steps by $2r+2$. The last three,
together with $\ell\ge\ell_-$ and $f\ge r/t$,
give the hypotheses of Lemma~\ref{lem:sharp} needed for
the two volume inequalities in
Theorem~\ref{thm:successive}.

For fixed $t$, $s$, $c_1$, and $c_2$, the five left-hand sides
in \eqref{eq:fivemargins} are affine functions of $r$. To
establish the inequalities for every $r\ge r_0(t)$, it is
therefore enough to check that each left-hand side is nonnegative
at $r=r_0(t)$ and has nonnegative slope.

\begin{theorem}\label{thm:tails}
Let $\C$ be an $[n,k]_q$ linear code with redundancy $\rho$,
let $4\le t\le\min\{k,\rho,31\}$, and set $r:=R_t(\C)$.
If $r\ge r_0(t)$, where $r_0(t)$ is given in
Table~\ref{tab:tails}, then $d_t(\C)\le2r+2$.
\end{theorem}

\begin{proof}
Suppose that $\C$ is a counterexample. Every value of $r_0(t)$
in Table~\ref{tab:tails} is greater than $2t$. Hence,
$r>2t$, and Lemma~\ref{lem:low-ratio} gives $\rho>2r$.
We may therefore apply Lemma~\ref{lem:sharp}.

For each $t$, the table specifies $s$, $(c_1,j_1)$, and
$(c_2,j_2)$. At $r=r_0(t)$, these choices satisfy the five
inequalities in \eqref{eq:fivemargins}, as well as
$\ell_-\ge j_1$ and $r/t\ge j_2$. The constants also satisfy
$c_i^{j_i}\ge3j_i$ and $c_i\ge(j_i+1)/j_i$ for $i=1,2$.

Each left-hand side in \eqref{eq:fivemargins} has nonnegative
slope as a function of $r$. Moreover, $\ell_-$ and $r/t$
increase with $r$. Thus, all these inequalities hold for every
$r\ge r_0(t)$.

The last three inequalities in \eqref{eq:fivemargins} now allow
us to apply Lemma~\ref{lem:sharp} to both puncturing steps.
As shown above, the resulting volume bounds give the
inequalities in Theorem~\ref{thm:successive}. The first two
inequalities in \eqref{eq:fivemargins} show that the two steps
puncture at most $2r+2$ coordinates. Theorem~\ref{thm:successive}
therefore gives $d_t(\C)\le2r+2$, contrary to the assumption
that $\C$ is a counterexample.
\end{proof}

\begin{table}[ht]
\centering
\begin{tabular}{rrrrrrr}
\hline
$t$ & $r_0(t)$ & $s$ & $c_1$ & $j_1$ & $c_2$ & $j_2$\\
\hline
4  & 102 & 2  & $9/8$ & 41 & $6/5$ & 24\\
5  & 39  & 1  & $4/3$ & 13 & $7/4$ & 5\\
6  & 28  & 2  & $3/2$ & 8  & $2$   & 4\\
7  & 30  & 2  & $7/4$ & 5  & $2$   & 4\\
8  & 26  & 3  & $7/4$ & 5  & $5/2$ & 2\\
9  & 28  & 4  & $7/4$ & 5  & $5/2$ & 2\\
10 & 31  & 4  & $7/4$ & 5  & $5/2$ & 2\\
11 & 33  & 5  & $7/4$ & 5  & $5/2$ & 2\\
12 & 37  & 5  & $7/4$ & 5  & $5/2$ & 2\\
13 & 38  & 6  & $7/4$ & 5  & $5/2$ & 2\\
14 & 42  & 6  & $7/4$ & 5  & $5/2$ & 2\\
15 & 44  & 7  & $7/4$ & 5  & $5/2$ & 2\\
16 & 48  & 7  & $7/4$ & 5  & $5/2$ & 2\\
17 & 49  & 8  & $7/4$ & 5  & $5/2$ & 2\\
18 & 53  & 8  & $7/4$ & 5  & $5/2$ & 2\\
19 & 55  & 9  & $7/4$ & 5  & $5/2$ & 2\\
20 & 58  & 9  & $7/4$ & 5  & $5/2$ & 2\\
21 & 60  & 10 & $7/4$ & 5  & $5/2$ & 2\\
22 & 64  & 10 & $7/4$ & 5  & $5/2$ & 2\\
23 & 66  & 11 & $7/4$ & 5  & $5/2$ & 2\\
24 & 69  & 11 & $7/4$ & 5  & $5/2$ & 2\\
25 & 71  & 12 & $7/4$ & 5  & $5/2$ & 2\\
26 & 74  & 13 & $7/4$ & 5  & $5/2$ & 2\\
27 & 77  & 13 & $7/4$ & 5  & $5/2$ & 2\\
28 & 79  & 14 & $7/4$ & 5  & $5/2$ & 2\\
29 & 82  & 14 & $7/4$ & 5  & $5/2$ & 2\\
30 & 85  & 15 & $7/4$ & 5  & $5/2$ & 2\\
31 & 88  & 15 & $7/4$ & 5  & $5/2$ & 2\\
\hline
\end{tabular}
\caption{Values used in the proof of Theorem~\ref{thm:tails}.
For each $t$, the inequalities checked at $r=r_0(t)$
continue to hold for every $r\ge r_0(t)$.}
\label{tab:tails}
\end{table}

\subsubsection{The case of large \texorpdfstring{$r$}{r} for \texorpdfstring{$t=3$}{t=3}}\label{sec:three}

Theorem~\ref{thm:tails} bounds the radius of a possible
counterexample for each $4\le t\le31$. We now obtain such a
bound for $t=3$. The nonbinary and binary cases require different
arguments, so we treat them separately.

\begin{theorem}\label{thm:ternary}
Let $\C$ be an $[n,k]_q$ linear code with redundancy $\rho$,
let $\min\{k,\rho\}\ge 3$, and set $r:=R_3(\C)$.
If $q=3$ and $r\ge48$, or if $q\ge4$ and $r\ge26$, then
$d_3(\C)\le2r+2$.
\end{theorem}

\begin{proof}
Suppose that $\C$ is a counterexample. In either range
$r>2t$, so Lemma~\ref{lem:low-ratio} gives $\rho>2r$.

We first choose the puncturing steps. If $r$ is odd, set
$\ell:=(r+1)/2$ and apply Theorem~\ref{thm:successive}
in one step, with order $3$ and radius $\ell$. This step
punctures $4\ell=2r+2$ coordinates. Moreover,
$4\ell-2r=2$, so the exponent in
\eqref{eq:simpleexponent} is sufficient for the volume
inequality in Theorem~\ref{thm:successive}.

If $r$ is even, set $\ell:=r/2$ and
$f:=\lceil r/3\rceil$. We use a first step of order $1$
and radius $\ell$, followed by a step of order $2$ and
radius $f$. The first step punctures $2\ell=r$
coordinates, and the two steps together puncture at most
$r+3f\le2r+2$ coordinates. For the first step, the
exponent in \eqref{eq:simpleexponent} is sufficient with
equality. For the second step it is sufficient that
$4f\ge r+2$, which holds in the range considered here.
We also have $3\ell\ge r$ and $3f\ge r$.

It remains to ensure that Lemma~\ref{lem:sharp} gives
\eqref{eq:simpleexponent} for these radii. For $q_0=3$,
take
$(c_1,j_1)=(6/5,24)$ and $(c_2,j_2)=(3/2,8)$.
For $q_0=4$, take
$(c_1,j_1)=(4/3,13)$ and $(c_2,j_2)=(7/4,5)$.
In each case, $c_i^{j_i}\ge3j_i$ and
$c_i\ge(j_i+1)/j_i$ for $i=1,2$.

Since $\ell\le(r+1)/2$ and, when $r$ is even,
$f\le(r+2)/3$, the conditions
$\beta_{q_0}(0,\ell)\ge c_1$ and
$\beta_{q_0}(r,f)\ge c_2$ follow by checking the
corresponding bounds at these upper estimates for the
radii. Indeed, $\beta_{q_0}(W,E)\ge c$ is equivalent to
$E\bigl(\beta_{q_0}(W,E)-c\bigr)\ge0$, and by \eqref{eq:beta}
the left-hand side is decreasing in $E$ for fixed $W$. After clearing denominators, the expressions to
be checked are affine in $r$. For $q_0=3$, both are
positive at $r=48$ and have positive slopes. For
$q_0=4$, the same is true at $r=26$.

Finally, $\ell\ge r/2\ge j_1$ and, when the second
step is used, $f\ge r/3\ge j_2$. Lemma~\ref{lem:sharp}
therefore supplies the required volume inequalities.
Its last assertion extends the argument with $q_0=4$
to every $q\ge4$. Theorem~\ref{thm:successive} gives
$d_3(\C)\le2r+2$, a contradiction.
\end{proof}

We now consider the binary field. We show that, when $r$ is
sufficiently large, a Hamming ball of radius $\lceil r/2\rceil$
contains enough vectors to apply Lemma~\ref{lem:hamming-ball}.

\begin{theorem}\label{thm:binarytail}
Let $\C$ be an $[n,k]_2$ linear code with redundancy $\rho$,
let $\min\{k,\rho\}\ge 3$, and set $r:=R_3(\C)$.
If $r\ge192$, then $d_3(\C)\le2r+2$.
\end{theorem}

\begin{proof}
Suppose that $\C$ is a counterexample. By
Lemma~\ref{lem:low-ratio}, $\lambda:=\rho/r>2$.
Set $\ell:=\lceil r/2\rceil$. We will prove that
$\vol_2(n,\ell)>2^{\rho+2}$. Lemma~\ref{lem:hamming-ball},
applied with $s=3$, will then give
$d_3(\C)\le4\ell\le2r+2$.

Eqs.~\eqref{eq:rate} and~\eqref{eq:coverrough} give
$$
n>\frac52r\lambda
\qquad\text{and}\qquad
n>\frac r{\mathrm e}2^{3(\lambda-1)}.
$$
We consider two ranges for $\lambda$.

Suppose first that $2<\lambda\le7/3$.
Lemma~\ref{lem:factorial}, $\mathrm e>19/7$, $n-\ell+1>n-\ell$ and
$\ell\le(r+1)/2$ give
$$
\vol_2(n,\ell)>
\frac{A_r(\lambda)^\ell}{3\ell},
\qquad
A_r(\lambda):=
\frac{19}{7}\left(\frac{5r\lambda}{r+1}-1\right).
$$
For $r\ge192$, we have $A_r(2)\ge24$ and
$A_r(7/3)\ge28$. The function
$\log A_r(\lambda)-2\lambda\log2$ is concave, so
its minimum on $[2,7/3]$ occurs at one of these
two values. Comparing them gives
$$
\log A_r(\lambda)-2\lambda\log2
\ge\frac13\log\frac{343}{256}.
$$
Since $\ell\ge r/2$ and $\rho=r\lambda$, it follows that
\begin{equation}\label{eq:binarylow}
\vol_2(n,\ell)>
\frac{2^\rho}{3\ell}
\left(\frac{343}{256}\right)^{r/6}.
\end{equation}

Suppose now that $\lambda\ge7/3$. By
\eqref{eq:rate}, $n>5\rho/2>5r$; in particular $n\ge\ell$.
The product formula,
Lemma~\ref{lem:factorial}, and
\eqref{eq:coverrough} give
$$
\vol_2(n,\ell)>
\frac1{3\ell}
\left(\frac r\ell-\frac14\right)^\ell
2^{3\ell(\lambda-1)}.
$$
Here we use
$\mathrm e(\ell-1)/\ell<4\le2^{3(\lambda-1)}/4$.
Since $r\ge192$ and $\ell\le(r+1)/2$, we have
$$
\frac r\ell-\frac14
\ge\frac{2r}{r+1}-\frac14
\ge\frac{1343}{772}>\sqrt3.
$$
Using $\ell\ge r/2$ and $\lambda\ge7/3$, we obtain
\begin{equation}\label{eq:binaryhigh}
\vol_2(n,\ell)>
\frac{3^{r/4}2^{3r(\lambda-1)/2}}{3\ell}
\ge\frac{2^\rho}{3\ell}
\left(\frac{27}{16}\right)^{r/12}.
\end{equation}

For every $r\ge192$, we have
$$
\left(\frac{343}{256}\right)^{r/6}>12r,
\qquad
\left(\frac{27}{16}\right)^{r/12}>12r.
$$
At $r=192$, these are the exact comparisons
$(343/256)^{32}>2304$ and $(27/16)^{16}>2304$.
They continue to hold as $r$ increases: for $A>1$,
the function $A^{r/d}/r$ is increasing when
$\log A/d>1/r$, and $\log A>(A-1)/A$ gives
$$
\frac{87}{6\cdot343}>\frac1{192},
\qquad
\frac{11}{12\cdot27}>\frac1{192}.
$$

Finally, $\ell\le r$. Thus, either
\eqref{eq:binarylow} or \eqref{eq:binaryhigh} gives
$\vol_2(n,\ell)>4\cdot2^\rho=2^{\rho+2}$,
contradicting the assumption that $\C$ is a counterexample.
\end{proof}

We have obtained an upper bound on $r$ for every order
$3\le t\le31$, while Theorem~\ref{thm:highorder} covers
$t\ge32$. Thus, $t$ and $r$ range over finite sets for a
possible counterexample. Theorem~\ref{thm:tail} then bounds
$\rho$, and Proposition~\ref{prop:inputs} gives $q<r$.
Consequently, only finitely many parameter tuples
$(q,\rho,t,r)$ remain.

\subsection{The remaining parameter set}

All unbounded ranges in $t$, $r$, and $\rho$ have now been excluded.
We next collect the restrictions that a counterexample must still satisfy.
Before describing the final finite parameter set, we remove two additional
ranges that are convenient to treat separately.

\begin{theorem}\label{thm:extensions}
Let $\C$ be an $[n,k]_q$ linear code with redundancy $\rho$,
let $t\in[\min\{k,\rho\}]$, and set $r:=R_t(\C)$.
If $r\le2t$ or $r-t\le15$, then $d_t(\C)\le2r+2$.
\end{theorem}

\begin{proof}
Suppose that $\C$ is a counterexample satisfying one of the
conditions in the statement. Theorem~\ref{thm:highorder}
excludes $t\ge32$, and Proposition~\ref{prop:inputs}
excludes $t\le2$. Thus, $3\le t\le31$.
Theorem~\ref{thm:finite} gives $\rho\ge51$,
Theorem~\ref{thm:tail} gives $\rho<D(t,r)$, and
Proposition~\ref{prop:inputs} gives $q<r$ and
$3r<2\rho$.

By Proposition~\ref{prop:inputs}, we also have $r\ge t+1$.
Since $r\le2t$ or $r-t\le15$, every such counterexample
has parameters in the following range:
\begin{equation}\label{eq:boundary-range}
\begin{gathered}
3\le t\le31,\qquad
t+1\le r\le\max\{2t,t+15\},\qquad
q<r\text{ a prime power},\\
\max\left\{
51,\,
2r-t+3,\,
\left\lfloor\frac{3r}{2}\right\rfloor+1
\right\}
\le\rho<D(t,r).
\end{gathered}
\end{equation}
The bound $\rho\ge2r-t+3$ follows from
\eqref{eq:failure-range}, and
$\rho\ge\lfloor3r/2\rfloor+1$ follows from $3r<2\rho$.

The range in \eqref{eq:boundary-range} contains $619{,}126$
tuples. The program \path{verify_extensions.py}, described in
Appendix~\ref{app:computations}, enumerates the range and
checks each tuple. In $154{,}813$ cases, a proved lower bound
on $n$ exceeds $M(q,\rho,t,r)$, contradicting
\eqref{eq:length-upper}. In $158{,}985$ cases, the conditions of
Lemma~\ref{lem:app-fixed-radius} hold at the test length $N_0$ of
Appendix~\ref{app:test-length}, and Lemma~\ref{lem:shorten} gives
$d_t(\C)\le2r+2$.
For the remaining $305{,}328$ tuples, the program checks
a sequence satisfying Theorem~\ref{thm:successive} at a test length, and Lemma~\ref{lem:shorten} again gives the same bound. The test lengths and the inequalities
checked in these cases are described in
Appendices~\ref{app:test-length}--\ref{app:steps}.

Thus, no tuple in \eqref{eq:boundary-range} can correspond
to a counterexample.
\end{proof}

We can now describe the parameters that remain. For $t=3$,
Theorems~\ref{thm:ternary} and~\ref{thm:binarytail}
give $r<192$ when $q=2$, $r<48$ when $q=3$, and
$r<26$ when $q\ge4$. For $4\le t\le31$,
Theorem~\ref{thm:tails} gives $r<r_0(t)$, where
$r_0(t)$ is listed in Table~\ref{tab:tails}.

Together with the preceding reductions, these results show
that every counterexample must satisfy
\begin{equation}\label{eq:finite}
\begin{gathered}
3\le t\le31,\qquad
r\ge\max\{2t+1,t+16\},\qquad
q<r\text{ a prime power},\\
r<
\begin{cases}
192, & t=3,\ q=2,\\
48, & t=3,\ q=3,\\
26, & t=3,\ q\ge4,\\
r_0(t), & 4\le t\le31,
\end{cases}
\qquad
\rho_0(q,t,r)\le\rho<D(t,r),
\end{gathered}
\end{equation}
where
$$
\rho_0(q,t,r):=
\begin{cases}
\max\{51,2r+4\}, & t=3,\ q=2,\\
\max\{51,2r+3\}, & t=3,\ q\ge3,\\
\max\{51,2r+1\}, & t\ge4.
\end{cases}
$$

For convenience, we summarize the origin of these restrictions.
\begin{center}
\small
\begin{tabular}{ll}
\hline
Restriction & Reason\\
\hline
$t\le31$
& Theorem~\ref{thm:highorder}\\
$r>2t$ and $r-t\ge16$
& Theorem~\ref{thm:extensions}\\
$\rho\ge51$
& Theorem~\ref{thm:finite}\\
$\rho<D(t,r)$
& Theorem~\ref{thm:tail}\\
$q<r$
& Proposition~\ref{prop:inputs}\\
upper bounds on $r$ for $t=3$
& Theorems~\ref{thm:ternary} and~\ref{thm:binarytail}\\
$r<r_0(t)$ for $4\le t\le31$
& Theorem~\ref{thm:tails}\\
lower bounds on $\rho$ for $t=3$
& Theorem~\ref{thm:order3}\\
$\rho\ge2r+1$ for $t\ge4$
& Lemma~\ref{lem:low-ratio}\\
\hline
\end{tabular}
\end{center}

Only finitely many tuples $(q,\rho,t,r)$ satisfy
\eqref{eq:finite}. For each of them, we still have the
upper bound $n\le M(q,\rho,t,r)$ from
\eqref{eq:length-upper} and the lower bounds on $n$
from \eqref{eq:cover} and \eqref{eq:coverrough}.
We use these bounds and Theorem~\ref{thm:successive}
to exclude the remaining tuples.

\subsection{The remaining finite cases}

It remains to exclude the parameters in \eqref{eq:finite}.
We first treat many pairs $(t,r)$ by puncturing sequences
that work for every admissible $q$ and $\rho$. For the
remaining cases, we choose a length $N$ that every putative
counterexample must have or exceed. We then show either
that $N>M(q,\rho,t,r)$, contradicting
\eqref{eq:length-upper}, or that every $[N,N-\rho]_q$
code has $d_t\le2r+2$. In the latter case,
Lemma~\ref{lem:shorten} gives the same conclusion at
every possible original length $n\ge N$.

\begin{theorem}\label{thm:completion}
No tuple in \eqref{eq:finite} can correspond to a
counterexample.
\end{theorem}

\begin{proof}
We begin with all $400$ pairs $(t,r)$ in \eqref{eq:finite} with
$5\le t\le31$ and $54$ of the $82$ pairs with $t=4$. For each of
these pairs, $r\ge2t+1$, so Lemma~\ref{lem:low-ratio} gives $\rho>2r$
and Lemma~\ref{lem:K} is available. For each such pair $(t,r)$, the
program checks a sequence of steps in
Theorem~\ref{thm:successive} without enumerating
the admissible values of $q$ and $\rho$.

At a step, let $W$ be the number of coordinates already
punctured, let $S$ be the nullity already obtained, and
let $\ell$ and $s$ be the radius and the increase in
nullity for the next step. The program checks
$$
\ell t\ge r,\qquad
K_2(W,\ell)>0,\qquad
K_2(W,\ell)^\ell\ge\ell!,
$$
where $K_2$ is defined in \eqref{eq:Ksymbolic}, as well as
$$
s-1-W+S\le\ell(t+1)-2r.
$$
Lemma~\ref{lem:K}, applied with $q_0=2$, then gives
$$
\vol_q(n-W,\ell)>
q^{\rho-W+S+s-1},
$$
as required by Theorem~\ref{thm:successive}. The program
also checks that the steps reach nullity $t$ and puncture
at most $2r+2$ coordinates.

It remains to consider $t=3$ and the $28$ pairs with $t=4$ not
covered by these sequences, namely
$r\in\{20,\ldots,39,41,42,43,44,46,48,51,53\}$.
For these pairs, the program enumerates the prime powers
$q<r$ and the integers $\rho$ in \eqref{eq:finite}.

For each resulting tuple $(q,\rho,t,r)$, the program
chooses an integer $N$ and establishes $n\ge N$ for
every possible counterexample. It uses the lower bounds
in \eqref{eq:finite-parameters} or, when a larger $N$
is needed, checks
$$
\vol_{q^t}(N-1,r)<q^{t\rho}.
$$
In the latter case, \eqref{eq:cover} implies $n\ge N$.

If $N>M(q,\rho,t,r)$, this contradicts
\eqref{eq:length-upper}. Otherwise, the program checks
at length $N$ the volume inequalities and the remaining
conditions of Theorem~\ref{thm:successive}. These checks
give $d_t(\C_0)\le2r+2$ for every $[N,N-\rho]_q$
code $\C_0$. Lemma~\ref{lem:shorten} then gives
$d_t(\C)\le2r+2$ for every possible original length
$n\ge N$.

There are $65{,}360$ numerical tuples: $40{,}145$
with $t=3$ and $25{,}215$ with $t=4$. In $1{,}216$
cases the lower bound on $n$ exceeds $M(q,\rho,t,r)$.
The other $64{,}144$ cases are settled by puncturing
sequences. The program derives the required pairs and
tuples from \eqref{eq:finite} and checks that all are
accounted for. Every comparison is exact. Hence, no
tuple in \eqref{eq:finite} can correspond to a
counterexample.
\end{proof}

The programs and input records are described in
Appendix~\ref{app:computations}. From the accompanying
repository, the command \texttt{sage verify.sage} runs
all the checks and ends with
\texttt{ALL COMPUTATIONAL CHECKS PASSED}.

\subsection{Proof of the conjecture}

We can now complete the proof of Conjecture~\ref{conj:packing-covering}.
All infinite parameter ranges have been excluded in the preceding results,
and Theorem~\ref{thm:completion} eliminates the remaining finite set.

\MioTeorema*

\begin{proof}
Suppose that, for some admissible $t$, the code $\C$ is a counterexample, and
set
$r:=R_t(\C).$

By Proposition~\ref{prop:inputs}, we may assume $t\ge3$, $r\ge t+1$, $q<r$, and the restrictions in Eq.~\eqref{eq:rate} hold. Moreover,
Theorem~\ref{thm:finite} excludes $\rho\le50$, while
Theorem~\ref{thm:tail} excludes $\rho\ge D(t,r)$.

Theorem~\ref{thm:highorder} excludes $t\ge32$. Hence, $3\le t\le31$. Theorem~\ref{thm:extensions} further excludes both
$r\le2t$ and $r-t\le15$.
Thus, $r\ge\max\{2t+1,t+16\}$.

If $t=3$, Theorem~\ref{thm:order3} gives the field-dependent lower bound on
$\rho$ appearing in Eq.~\eqref{eq:finite}. If $4\le t\le31$, then
$r>2t$, so Lemma~\ref{lem:low-ratio} gives $\rho\ge2r+1$.

Finally, Theorem~\ref{thm:tails} excludes
$r\ge r_0(t)$ for $4\le t\le31$, while for $t=3$,
Theorems~\ref{thm:ternary} and~\ref{thm:binarytail} exclude $r\ge48$ for $q=3$, $r\ge26$ for $q\ge4$, and $r\ge192$ for $q=2$.

Consequently, every remaining counterexample would have parameters satisfying
Eq.~\eqref{eq:finite}. But Theorem~\ref{thm:completion} shows that no tuple
in that finite set can correspond to a counterexample. This contradiction
proves $d_t(\C)\le2R_t(\C)+2$.
\end{proof}

\begin{remark}
As observed in~\cite{YS}, generalized Hamming weights are much better understood than generalized covering radii, and Theorem~\ref{thm:main} makes the known results on the
former available for the latter. For instance, since the weight
hierarchy is strictly increasing~\cite{Wei}, we have
$d_t(\C)\ge d(\C)+t-1$ and hence
$$
R_t(\C)\ge\left\lfloor\frac{d(\C)+t-2}{2}\right\rfloor .
$$
Similarly, the explicit weight hierarchies known for cyclic codes, $q$-ary Reed--Muller
codes, affine Cartesian codes and other families
\cite{FTW,HP,BD,nardi2026structure,moreno2026generalized} give explicit
lower bounds on their generalized covering radii.
\end{remark}

\section*{Declaration on the use of AI}
The overall strategy of the proof was proposed by the authors. In particular, the authors developed the idea of translating the conjecture into a bound on the number of columns of a parity-check matrix contained in subspaces, and of comparing the resulting upper bounds on the length with the lower bounds coming from the generalized covering radius.
Through an iterative discussion, the authors used ChatGPT 6-Astra to develop this strategy into a complete argument, refine intermediate estimates, identify parameter ranges requiring computational treatment, organize and clarify the exposition, and assist in writing the search and verification programs described in Appendix~\ref{app:computations}.
The tool also assisted with the technical writing of
Section~\ref{sec:reduction}. All content, claims, and conclusions have been reviewed and verified by the authors, who take full responsibility for them.

\medskip

\section*{Acknowledgments}
This research has been partially supported by the Italian National Group for Algebraic and Geometric Structures and their Applications (GNSAGA -- INdAM), by Università Italo Francese (UIF/UFI) via PHC Galileo 2026 -- G26-260/54322VM and by the University of Naples Federico II, which funded the ``Naples--Rennes Agreement on Scientific Cooperation''.
G. N. Alfarano is supported by the Agence Nationale de la Recherche through grant number ANR-24-CPJ1-0075-01. A. Neri is supported by the University of Naples Federico II through the FRA2024 Project CUP E63C26001530001  ``NHACT -- New Horizons in Algebraic Coding Theory''.

\medskip

\bibliographystyle{abbrv}
\bibliography{references}

\appendix
\section{Bounds and computer calculations}\label{app:computations}

We first prove the generalized-weight bounds used in the finite cases and the step-count bound used in Proposition~\ref{prop:large-diagonal}. We
then explain the parameter ranges, the inequalities checked for each
tuple, and the accompanying programs. The programs enumerate the
required tuples independently of the input records. All proof
comparisons use integers or rational numbers. The programs
(Python~3, run through SageMath) and input data are available at
\url{https://github.com/gianiraalfarano/packing-covering-conjecture}.

\subsection{Generalized-weight bounds}\label{app:weights}
For some parameter tuples we need bounds on several generalized Hamming
weights at the same length. We give a recursion that applies to every code
with the prescribed length, field, and redundancy.

The following construction combines ordinary distance bounds with Lemma~\ref{lem:hamming-ball} in a finite recursive procedure.

For $N>h\ge0$, let $A_q(N,h)$ be the largest integer $d\in[1,h+1]$
satisfying all three necessary conditions below:
\begin{align}
\vol_q\bigl(N,\lfloor(d-1)/2\rfloor\bigr)&\le q^h,\label{eq:hamming}\\
\vol_2\bigl(N-1,(d-2)/2\bigr)&\le 2^{h-1}\qquad\text{if $q=2$ and $d$ is even},\label{eq:even}\\
\sum_{i=0}^{N-h-1}\left\lceil\frac{d}{q^i}\right\rceil&\le N.\label{eq:griesmer}
\end{align}
Inequality~\eqref{eq:even} follows by puncturing one coordinate of a binary code of minimum distance at least $d\ge2$. Dimension is preserved, redundancy becomes $h-1$, and the distance is at least $d-1$. Inequality~\eqref{eq:griesmer} is the ordinary Griesmer bound, valid also when the given length includes zero coordinates. These are necessary conditions only; passing them is not an existence assertion.
For completeness, the Griesmer inequality follows by induction on dimension.
Choose a minimum word $c$ of weight $d$. Puncturing its support has kernel
$\langle c\rangle$, by Lemma~\ref{lem:puncture}. Let $v\in\C$ be a codeword
whose image is a nonzero residual word of weight $w$, and average the weights of its $q$ lifts $v+\mu c$, $\mu\in\Fq$, all of which are nonzero codewords.
On $\supp(c)$ each coordinate is nonzero in exactly $q-1$ lifts, so
$d\le w+(q-1)d/q$. Thus, $w\ge\ceil{d/q}$, and induction gives
\eqref{eq:griesmer}, using
$\ceil{\ceil{d/q}/q^i}=\ceil{d/q^{i+1}}$.

For $s\in [N-h]$, set
$$
 L_q(N,h,s)=\min\Bigl(\{h+s\}\cup
 \{(s+1)\ell:\vol_q(N,\ell)>q^{h+s-1},\ \ell\ge0\}\Bigr).
$$
Define $B_1(q,N,h)=A_q(N,h)$ and, for $t\ge2$,
\begin{equation}\label{eq:recursion}
\begin{aligned}
 B_t(q,N,h)=\min\biggl\{L_q(N,h,t),\quad
 \min_{1\le s<t}\ \max_{s\le j\le B_s(q,N,h)}
 \bigl[j+B_{t-s}(q,N-j,h-j+s)\bigr]\biggr\}.
\end{aligned}
\end{equation}
The maximum is over integers. Every recursive order is smaller than $t$, so this is a finite definition. By construction, $B_s(q,N,h)\le h+s$ (the generalized Singleton bound), hence all residual redundancies in the formula are nonnegative. The residual dimension is always $N-h-s\ge t-s$.

\begin{theorem}\label{thm:recursive}
Every $[N,N-h]_q$ code satisfies $d_t\le B_t(q,N,h)$ for $t\in[N-h]$.
\end{theorem}
\begin{proof}
The assertion for $t=1$ follows from \eqref{eq:hamming}--\eqref{eq:griesmer}. Generalized Singleton and Lemma~\ref{lem:hamming-ball} give the bound $L_q$.

Fix $1\le s<t$. By induction, the actual value $j=d_s$ lies between $s$ and $B_s(q,N,h)$. Puncturing a minimum-support subcode and induction at order $t-s$ give
$$
 d_t\le j+B_{t-s}(q,N-j,h-j+s).
$$
Maximizing over every possible $j$ makes the bound independent of its unknown actual value. It holds for every $s$, so taking the minimum gives \eqref{eq:recursion}.
\end{proof}

\subsection{Step counts for successive puncturing}\label{app:iteration-counts}
We first give the fixed-radius form of Theorem~\ref{thm:successive} used by
the programs. We then bound the number of steps in the three-phase argument
of Proposition~\ref{prop:large-diagonal}.

For integers $\ell,A\ge1$ and $L\ge0$, write
$$
 G_\ell(L,A)=A\sum_{j=0}^{L-1}\ell^j+
 \ell\sum_{i=0}^{L-1}i\, \ell^{L-1-i},
$$
so that $G_\ell(0,A)=0$. The identity $G_\ell(L,A)=\ell G_\ell(L-1,A)+A+\ell(L-1)$ follows directly from
this expression.

\begin{lemma}\label{lem:app-fixed-radius}
Let $\C$ be an $[n,n-\rho]_q$ code, let $t\in[n-\rho]$, and let
$\ell,L,A\ge1$. Set $B=\ell(t+L)$. If $G_\ell(L,A)\ge t$, $n>B$, and
$\vol_q(n-B,\ell)>q^{\rho+A-1}$, then $d_t(\C)\le B$.
\end{lemma}

\begin{proof}
If $\rho+t\le B$, the generalized Singleton bound gives the claim.
Otherwise, start with $S_0=0$ and, until the accumulated nullity reaches
$t$, take
$$
s_i=\min\{t-S_{i-1},\ A+(\ell-1)S_{i-1}+\ell(i-1)\}
$$
and take the radius $\ell$ at every step. Before the last step, $S_i=\ell S_{i-1}+A+\ell(i-1)=G_\ell(i,A)$.
Thus, $G_\ell(L,A)\ge t$ ensures that at most $L$ steps are needed.
With $W_i=\ell(S_i+i)$, each step satisfies $W_i\le B$, and
$$
\rho-W_{i-1}+S_{i-1}+s_i-1\le\rho+A-1.
$$
The assumed ball inequality therefore implies
\eqref{eq:successive-uniform} at every step. Apply
Theorem~\ref{thm:successive}.
\end{proof}
In Proposition~\ref{prop:large-diagonal}, set $a:=r-t$ and
$c:=\max\{0,\lceil{(2a-t)/3}\rceil\}$. The integers $L_4,L_3,L_2$ are the
least nonnegative integers with
$$
 G_4(L_4,5)\ge c,\qquad G_3(L_3,4)\ge a-c,
 \qquad\text{and}\qquad G_2(L_2,3)\ge t-a.
$$
These integers bound the numbers of steps used in
Proposition~\ref{prop:large-diagonal}. Indeed, in each phase of the
rule~\eqref{eq:phase-rule}, the first admissible step size is at least $5$,
$4$ or $3$, respectively, and every step of radius $\ell$ increases the
admissible size by $(\ell-1)s+\ell$. Hence, as in the proof of
Lemma~\ref{lem:app-fixed-radius}, after $i$ steps of a phase the nullity
gained in that phase is at least $\min\{d,G_\ell(i,A_\ell)\}$ with
$(A_4,A_3,A_2)=(5,4,3)$, where $d$ is the nullity required in that phase.
Indeed, if $g_i$ denotes the nullity gained after $i$ uncapped steps and
$\sigma_1\ge A_\ell$ the first admissible size, then the $i$-th step adds
$\sigma_1+(\ell-1)g_{i-1}+\ell(i-1)$, so $g_i=\ell g_{i-1}+\sigma_1+\ell(i-1)\ge G_\ell(i,A_\ell)$
by induction.
The desired inequality is
\begin{equation}\label{eq:step-inequality}
 c+4L_4+3L_3+2L_2\le a+2.
\end{equation}
For $32\le t<64$ there are finitely many integers $t/3<a\le t$.
For each of them, computing $c$, $L_4$, $L_3$ and $L_2$ from the recurrence for $G_\ell$ gives nonnegative margin in
\eqref{eq:step-inequality}. The complete per-order minima are recorded in
\path{outputs/independent_audit.json} and recomputed by
\path{code/independent_audit.py}. Grouping consecutive orders, the minima
are shown in Table~\ref{tab:step-slack}.
\begin{table}[ht]
\centering
\begin{tabular}{rrrr}
\hline
Orders & Minimum margin  & Orders & Minimum margin \\\hline
32--35 & 1 & 48--51 & 5\\
36--39 & 1 & 52--55 & 5\\
40--43 & 1 & 56--59 & 5\\
44--47 & 3 & 60--63 & 6\\
\hline
\end{tabular}
\caption{Exact minimum of $a+2-c-4L_4-3L_3-2L_2$ over the indicated
orders and all integers $t/3<a\le t$.}\label{tab:step-slack}
\end{table}

For $t\ge64$ the following analytic bounds suffice. Since
$G_\ell(L,A)\ge A(\ell^L-1)/(\ell-1)$, if $a\le t/2$ then $c=L_4=0$ and
$$
 3L_3+2L_2\le
 F(t):=3\left\lceil\log_3\left(1+\frac t4\right)\right\rceil+
 2\left\lceil\log_2\left(1+\frac t3\right)\right\rceil
 \le\frac t3+2\le a+2.
$$
If $a>t/2$, then $c\le(2a-t+2)/3$ and
$$
\begin{aligned}
 4L_4+3L_3+2L_2\le J(t):={}&
 4\left\lceil\log_4\frac{t+7}{5}\right\rceil+
 3\left\lceil\log_3\frac{t+2}{2}\right\rceil+2\left\lceil\log_2\frac{t+6}{6}\right\rceil
 \le\frac t2+\frac43 < a-c+2.
\end{aligned}
$$
The two elementary inequalities for $F,J$ are checked at the integers
$64\le t<128$ by comparing powers, without logarithmic approximation.
They propagate to all larger integers: writing $T=2t$ or $2t+1$, each ceiling
increases by at most one. Hence, $F(T)\le F(t)+5$ and
$J(T)\le J(t)+9$, while the corresponding right sides increase by at least
$t/3$ and $t/2$. This proves \eqref{eq:step-inequality} for all $t\ge64$.

\subsection{The finite parameter ranges}\label{app:ranges}

We describe the ranges checked by the programs. Throughout this appendix,
$r=R_t(\C)$, and $D(t,r)$ and $M(q,\rho,t,r)$ are the bounds in
Theorem~\ref{thm:tail} and Eq.~\eqref{eq:M}. An input record specifies
necessary parameters for a counterexample; it does not assert the
existence of a code with those parameters.

The first program checks the tuples in Eq.~\eqref{eq:finite-parameters}
with $\rho\le50$. The second checks Eq.~\eqref{eq:gap-range}, where
$1\le r-t\le5$. These ranges contain $59{,}086$ and
$2{,}396{,}360$ tuples, respectively. The third program checks the
$619{,}126$ tuples in Eq.~\eqref{eq:boundary-range}. This is the
range used in the proof of Theorem~\ref{thm:extensions}; it has
$3\le t\le31$ and $t+1\le r\le\max\{2t,t+15\}$.

Finally, the fourth program checks the remaining range
Eq.~\eqref{eq:finite}. It excludes 454 pairs $(t,r)$ by symbolic
puncturing sequences, without enumerating $q$ and $\rho$ for those pairs.
Of these, 54 have $t=4$ and 400 have $5\le t\le31$.
The other cases have $t=3$ or $t=4$ and comprise 65,360 numerical tuples.
The upper bounds on $r$ in Eq.~\eqref{eq:finite}, together with
$\rho_0(q,t,r)\le \rho<D(t,r)$, make this last enumeration finite.

The programs derive the expected tuples from these displayed ranges.
They reject a missing or repeated tuple rather than relying on the
number of records supplied in an input file. The ranges may overlap:
this does not affect the argument, since each program verifies all the
tuples required for its own result.

\subsection{Lower bounds on the test length}\label{app:test-length}

For a fixed tuple $(q,\rho,t,r)$, the programs choose an integer
$N\ge\rho+t$ and prove $n\ge N$ for any counterexample with these
parameters. If $N>M(q,\rho,t,r)$, Eq.~\eqref{eq:length-upper} excludes
the tuple. Otherwise, it suffices to prove $d_t\le2r+2$ for every
$[N,N-\rho]_q$ code. Lemma~\ref{lem:shorten} then gives the same
bound for every original length $n\ge N$. Notice that the shortened
code need not have generalized covering radius $r$.

The elementary lower bounds in Eq.~\eqref{eq:finite-parameters} give
$$
n\ge\rho+5t-1,
\qquad
n\ge\left\lfloor\frac{5\rho}{2}\right\rfloor+1.
$$
The condition $N\ge\rho+t$ is also imposed. When these bounds do not
suffice, a proposed test length is checked by the strict inequality
$$
\vol_{q^t}(N-1,r)<q^{t\rho}.
$$
Indeed, Eq.~\eqref{eq:cover} gives
$\vol_{q^t}(n,r)\ge q^{t\rho}$. Since the volume is increasing in its
length, the displayed strict inequality implies $n\ge N$.

Two faster sufficient tests are used in the fixed-gap and boundary
calculations. Set $m:=2r-t+2$ and let $a:=r-t$. The inequality
Eq.~\eqref{eq:coverrough}, together with $\mathrm e<3$, gives
$$
n>\frac r3 q^{t(\rho-r)/r}
\ge\frac r3 q^{\lfloor t(\rho-r)/r\rfloor}.
$$
Hence, the integer
$$
N_0:=\max\left\{
\rho+5t-1,
\left\lfloor\frac{5\rho}{2}\right\rfloor+1,
\left\lfloor\frac r3q^{\lfloor t(\rho-r)/r\rfloor}\right\rfloor+1
\right\}
$$
is a valid lower bound on $n$. The program compares $N_0$ with the
rank-$x$ incidence bound, where $x=m-b$ and $b$ is the least positive
integer with $q^b\ge t$, and with the weighted-line bound. The minimum
of these upper bounds is at least $M(q,\rho,t,r)$, and each bound
separately applies to a counterexample. Thus, $N_0$ exceeding that
minimum is a contradiction.

If this direct comparison is insufficient, a radius-two calculation
gives the result in many further cases. Take $L=a+1$, $B=2(t+L)=2r+2$, and
$$
A:=\max\left\{1,
\left\lceil\frac{t-2(2^L-L-1)}{2^L-1}\right\rceil
\right\}.
$$
Since $G_2(L,A)=A(2^L-1)+2(2^L-L-1)\ge t$, the conditions
$N_0>B$ and
$\vol_q(N_0-B,2)>q^{\rho+A-1}$ imply, by
Lemma~\ref{lem:app-fixed-radius} with $\ell=2$, that every $[N_0,N_0-\rho]_q$ code
satisfies $d_t\le B$, and Lemma~\ref{lem:shorten} transfers this to $\C$.
All powers and comparisons here are evaluated as integers.

For the tuples that survive both direct tests, the fixed-gap and
boundary files can supply a sharper length $N$. A sufficient check is
$$
(N-1)^r<r!q^{t(\rho-r)}.
$$
To see this, suppose that $n\le N-1$. Eq.~\eqref{eq:coverrough}
gives $\binom nr\ge q^{t(\rho-r)}$, whereas
$\binom nr\le n^r/r!\le(N-1)^r/r!$, a contradiction.
Some boundary records use the exact covering-volume comparison instead.
The programs distinguish the two cases and check the indicated strict
inequality. No floating-point approximation of a root is used.

\subsection{Checking the generalized-weight bounds}\label{app:records}

A generalized-weight record has parameters $(q,N,h,j;b)$ and claims
that every $[N,N-h]_q$ code satisfies $d_j\le b$. The program first
checks $h\ge0$, $1\le j\le N-h$, and $j\le b\le h+j$. It then
checks one of the following justifications.
\begin{enumerate}
\item The generalized Singleton bound gives $b=h+j$.
\item If $j=1$, every larger proposed minimum distance
$b<d\le h+1$ violates at least one of the necessary conditions
\eqref{eq:hamming}--\eqref{eq:griesmer}.
\item An integer $\ell\ge0$ satisfies
$\vol_q(N,\ell)>q^{h+j-1}$ and $(j+1)\ell\le b$.
Lemma~\ref{lem:hamming-ball} then gives the claim.
\item Integers $\ell,L,A\ge1$ satisfy
$G_\ell(L,A)\ge j$, $N>\ell(j+L)$, and
$\vol_q(N-\ell(j+L),\ell)>q^{h+A-1}$, with $\ell(j+L)\le b$.
Lemma~\ref{lem:app-fixed-radius} applies.
\item For an integer $1\le s<j$, a preceding record gives
$d_s\le b_s$. For every possible integer $s\le z\le b_s$, a
preceding record gives a bound $d_{j-s}\le b_z$ at
$(q,N-z,h-z+s)$, and $z+b_z\le b$. If the actual value of
$d_s$ is $z$, Lemma~\ref{lem:puncture} gives the asserted bound.
The program checks a child for every such $z$.
\end{enumerate}
The records are read in order, so a child bound must already have been
checked when it is used. The program rejects duplicate records and
missing roots. A second program checks that every retained record is
reachable from a root. The recursion at the beginning of the appendix
explains why these rules give bounds that are uniform over all codes
with the stated parameters; for rule~4, the justification is given
directly by Lemma~\ref{lem:app-fixed-radius}.

For example, the tuple $(q,\rho,t,r)=(2,10,5,6)$ has a verified lower
bound $n\ge34$ and length upper bound $M(2,10,5,6)=21$; it is excluded
without a weight record. For $(2,16,3,6)$, the covering bound \eqref{eq:cover} gives $n\ge112$,
while $M(2,16,3,6)=137$, so the length bounds are inconclusive. The
record uses the test length $N=96$, justified by
$\vol_8(95,6)<8^{16}$. Since $\vol_2(96,3)=147{,}537>2^{17}$,
Lemma~\ref{lem:hamming-ball} gives $d_2\le9$ for every $[96,80]_2$
code. For each possible value $z=2,\ldots,9$ of $d_2$, puncturing a
minimum-support subcode leaves a $[96-z,78]_2$ code, and the ordinary
bounds \eqref{eq:hamming}--\eqref{eq:griesmer} give $d_1\le6,6,6,6,5,4,4,4$,
respectively. In every case $z+d_1\le13$, so Lemma~\ref{lem:puncture}
gives $d_3\le13\le2r+2$ at length $96$, and Lemma~\ref{lem:shorten}
covers every larger length. The file \path{example.sage} displays
every step of this argument.

\subsection{Checking the puncturing steps}\label{app:steps}

For a numerical record, write $(\ell_i,s_i)$ for its successive Hamming
radii and increases of nullity, and set $B:=2r+2$. Starting with
$S_0=W_0=0$, the program checks, for each step,
$$
1\le s_i\le t-S_{i-1},
\qquad
\vol_q(N-W_{i-1},\ell_i)>
q^{\rho-W_{i-1}+S_{i-1}+s_i-1},
$$
together with the required residual dimension. It then sets
$$
S_i=S_{i-1}+s_i,
\qquad
W_i=W_{i-1}+\ell_i(s_i+1).
$$
It requires $W_i\le B$ at every step, $S_v=t$ at the end, and
$W_v-t<\rho$ when padding to the prescribed support size is used.
These are the hypotheses of Theorem~\ref{thm:successive} at the
fixed test length $N$. The strict inequalities imply
$d_t\le B$ for every $[N,N-\rho]_q$ code, and
Lemma~\ref{lem:shorten} transfers the bound to $n\ge N$.

The $1{,}022$ symbolic sequences in the high-order argument cover
$6\le a\le9$ and $32\le t<5a^2$, where $a=r-t$. The program
constructs the expected set of all such pairs $(t,a)$ and checks
Eq.~\eqref{eq:symbolic-chain} at every step. In particular, it
clears the denominator $5a$ in the last inequality and checks
$$
5a(s_i-1-W_{i-1}+S_{i-1})
\le5a\bigl(\ell_i(t+1)-2r\bigr)+2(\ell_it-r).
$$
The conditions $\ell_it\ge r$ and $12\ell_i\le r$ allow the use of
Lemma~\ref{lem:ball-exponent}. The final checks $S_v=t$ and
$W_v\le2r+2$ then give the conclusion by
Theorem~\ref{thm:successive}.

In the final range \eqref{eq:finite}, $454$ pairs $(t,r)$
are excluded independently of $q$ and $\rho$. For each step of a
recorded sequence, the program checks $\ell_it\ge r$ and, with
$K_2(W_{i-1},\ell_i)$ as in Eq.~\eqref{eq:Ksymbolic}, checks
$$
K_2(W_{i-1},\ell_i)>0,
\qquad
K_2(W_{i-1},\ell_i)^{\ell_i}\ge \ell_i!,
$$
and
$$
s_i-1-W_{i-1}+S_{i-1}\le \ell_i(t+1)-2r.
$$
Lemma~\ref{lem:K}, with $q_0=2$, supplies the strict ball inequality
needed for the step. The same final nullity and
support checks finish the argument. The records cover all $400$ remaining
pairs $(t,r)$ with $5\le t\le31$ and $54$ pairs with $t=4$. The
remaining pairs with $t=4$, and those with $t=3$, are checked
numerically as described next.

For a numerical tuple in \eqref{eq:finite}, the input gives
an integer $N\ge\max\{\rho+5t-1,\lfloor5\rho/2\rfloor+1\}$.
If $N$ exceeds this elementary lower bound, the program checks
$\vol_{q^t}(N-1,r)<q^{t\rho}$. It then checks either
$N>M(q,\rho,t,r)$ or every step of a puncturing sequence at $N$.
There are $1{,}216$ length contradictions and $64{,}144$
puncturing arguments, distributed as follows.
\begin{center}
\begin{tabular}{rrrr}
\hline
$t$ & Tuples & Length & Puncturing\\
\hline
$3$ & $40{,}145$ & $314$ & $39{,}831$\\
$4$ & $25{,}215$ & $902$ & $24{,}313$\\
\hline
\end{tabular}
\end{center}
The totals are checked against an independently enumerated parameter
set. A stored count cannot replace a missing tuple or an inequality
that fails.

The infinite-radius calculations for $4\le t\le31$ use the choices
in Table~\ref{tab:tails}. The program checks the prefactor
conditions in Eq.~\eqref{eq:cj} and the five expressions in
Eq.~\eqref{eq:fivemargins}. For each expression $f(r)=\alpha r+\beta$,
it checks $f(r_0(t))\ge0$ and
$f(r_0(t)+1)-f(r_0(t))=\alpha\ge0$. This establishes the five
inequalities for every integer $r\ge r_0(t)$, rather than at a
finite sample of radii. The order-three endpoint and propagation
comparisons in the proofs of Theorems~\ref{thm:ternary} and
\ref{thm:binarytail} are also checked with exact rational or integer
arithmetic.

\subsection{Programs and input files}\label{app:programs}

The programs and input records are available at
\url{https://github.com/gianiraalfarano/packing-covering-conjecture}.
They are run with SageMath~10 or later. From the repository directory, the
complete verification is run with
\begin{verbatim}
sage verify.sage
\end{verbatim}
This command unpacks the input records from \path{proof_data.zip},
runs the five checking programs described below, and ends with
\texttt{ALL COMPUTATIONAL CHECKS PASSED} if every check succeeds. With
SageMath~10.8, the complete run takes about two minutes on a standard
computer.

The programs \path{verify_range.py}, \path{verify_fixed_gaps.py},
\path{verify_extensions.py} and \path{verify_new.py} treat, respectively,
the ranges $\rho\le50$, $1\le r-t\le5$, Eq.~\eqref{eq:boundary-range},
and Eq.~\eqref{eq:finite} together with the large-radius bounds of
Section~\ref{sec:low}. Each program enumerates its parameter range
independently of the input records and checks, with exact integer or
rational arithmetic, every inequality described in
Appendices~\ref{app:test-length}--\ref{app:steps}. The fifth program,
\path{independent_audit.py}, was written separately and uses none of the
others. It rechecks the parameter totals, the puncturing sequences, the
symbolic and large-radius inequalities, and the step counts of
Appendix~\ref{app:iteration-counts}, and it verifies that every
generalized-weight record is used.

As an illustration, \path{example.sage} carries out the length
and generalized-weight calculations for $(q,\rho,t,r)=(2,16,3,6)$ discussed in
Appendix~\ref{app:records}. The file \path{docs/verification.md}
describes the format of the input records. No linear codes are
enumerated in any of these calculations.

\end{document}